\documentclass[letterpaper,twocolumn,american,showpacs,pr,aps,superscriptaddress,eqsecnum,nofootinbib]{revtex4}

\usepackage{amsfonts}
\usepackage{amsmath}
\usepackage{amssymb}
\usepackage{graphicx}%
\usepackage{color}
\providecommand{\U}[1]{\protect\rule{.1in}{.1in}}
\newtheorem{theorem}{Theorem}

\newtheorem{corollary}[theorem]{Corollary}

\newtheorem{definition}[theorem]{Definition}
\newtheorem{example}[theorem]{Example}

\newtheorem{lemma}[theorem]{Lemma}

\newtheorem{proposition}[theorem]{Proposition}
\newtheorem{remark}[theorem]{Remark}

\newenvironment{proof}[1][Proof]{\noindent\textbf{#1.} }{\ \rule{0.5em}{0.5em}}

\newcommand*{\ket}
[1]{\mathopen{|}#1\mathclose{\rangle}}
\definecolor{nblue}{rgb}{0.2,0.2,0.7}
\definecolor{ngreen}{rgb}{0.2,0.6,0.2}
\definecolor{nred}{rgb}{0.7,0.2,0.2}
\definecolor{nblack}{rgb}{0,0,0}

\newcommand{\blk}{\color{nblack}}

\begin{document}
%
%
%
%
%
%

\title{The theory of manipulations of pure state asymmetry I: basic tools, equivalence classes, and single copy transformations}
\author{Iman Marvian}
\affiliation{Perimeter Institute for Theoretical Physics, 31 Caroline St. N, Waterloo, \\
Ontario, Canada N2L 2Y5}
\affiliation{Institute for Quantum Computing, University of Waterloo, 200 University Ave. W, Waterloo, Ontario, Canada N2L 3G1}

\author{Robert W. Spekkens}
\affiliation{Perimeter Institute for Theoretical Physics, 31 Caroline St. N, Waterloo, \\
Ontario, Canada N2L 2Y5}
\date{March 18, 2011}

\begin{abstract}

 If a system undergoes symmetric dynamics, then the final state of the system can only break the symmetry in ways in which it was broken by the initial state, and its measure of asymmetry can be no greater than that of the initial state. It follows that for the purpose of understanding the consequences of symmetries of dynamics, in particular, complicated and open-system dynamics, it is useful to introduce the notion of a state's \emph{asymmetry properties}, which includes the type and measure of its asymmetry.  We demonstrate and exploit the fact that the asymmetry properties of a state can also be understood in terms of information-theoretic concepts, for instance in terms of the state's ability to encode information about an element of the symmetry group. We show that the asymmetry properties of a pure state $\psi$ relative to the symmetry group $G$ are completely specified by the characteristic function of the state, defined as $\chi_{\psi}(g)\equiv \langle \psi|U(g)|\psi\rangle$ where $g\in G$ and $U$ is the unitary representation of interest. For a symmetry described by a compact Lie group $G$, we show that two pure states can be reversibly interconverted one to the other by symmetric operations if and only if their characteristic functions are equal up to a 1-dimensional representation of the group.  Characteristic functions also allow us to easily identify the conditions for one pure state to be converted to another by symmetric operations (in general irreversibly) for the various paradigms of single-copy transformations: deterministic, state-to-ensemble, stochastic and catalyzed.
\end{abstract}

\maketitle

\tableofcontents

\section{Introduction}

Symmetry arguments are ubiquitous in physics.  Their prominence stems from the fact that for many systems of interest, the dynamics are sufficiently complicated that one cannot hope to characterize their evolution completely, whereas by appealing to the symmetries of the dynamical laws one can easily infer many useful results.
One of the best known examples of such a result is Noether's theorem, according to which a differentiable symmetry of the Hamiltonian or action entails a conservation law (See, e.g. \cite{Goldstein}).  But there are innumerable results of this sort; symmetry arguments have broad applicability across many fields of physics.

We are interested in determining all the consequences of a symmetry of the dynamics in quantum theory. To find these consequences we ask the following question: Given two quantum states, $\rho$ and $\sigma$, does there exist a time evolution with the given symmetry such that under this time evolution the first state evolves to the second?

Suppose, for instance, that the symmetry under consideration is rotational symmetry. Clearly, rotationally-invariant time evolutions cannot take a rotationally-symmetric state to one that breaks the rotational symmetry.  So to answer these types of questions we need to know the extent to which each of the two states breaks the rotational symmetry. It is intuitively clear that there are many different ways in which a quantum state may be asymmetric.  For instance, consider a spin-1/2 particle with spin in the $\hat{z}$ direction and another with spin in the $\hat{x}$ direction.  Neither is invariant under the full rotation group, but because they point in different directions, they break the rotational symmetry differently.  Furthermore, it is intuitively clear that asymmetry must be quantifiable.  For instance, the precision with which one can specify a direction in space, a measure of rotational asymmetry, varies with the quantum state one uses to do so.

We will say that two states have exactly the same \emph{asymmetry properties} (with respect to a given symmetry group) if there exists a symmetric time evolution which transforms the first state to the second and a symmetric time evolution which transforms the second state to the first.  Thus, the symmetric operations define equivalence classes of states and the asymmetry properties of a state are precisely those that are necessary and sufficient to determine its equivalence class.  If the symmetry in question is associated with a representation of the group G, we call the equivalence relation \emph{G-equivalence}.  We will consider G-equivalence classes of pure states for the case of arbitrary compact Lie groups and finite groups.


The above definition of asymmetry properties  is based on the intuition that asymmetry is something which cannot be generated by symmetric time evolutions.  We call this the \emph{constrained-dynamical} perspective.

 However, one can also take an \emph{information-theoretic} perspective on how to define the asymmetry properties of a state.  Recall that a quantum state breaks a symmetry, say rotational symmetry, if for some non-trivial rotations, the rotated version of the state is not the same as the state itself, i.e. they are distinguishable.   In this case, the ensemble of states corresponding to the orbit of the state under rotations can act as an encoding when the message to be encoded is an element of the rotation group.


To understand better the information-theoretic point of view, consider the following scenario:  Suppose Alice  wants to inform Bob about  a randomly chosen direction in space. She can  prepare a quantum system specifying the direction and send it to Bob. For example, to send a direction in a plane she may prepare a number of photons  polarized  in that direction. Clearly to transmit more information about this direction, Alice should prepare the quantum system in a state which sharply specifies the chosen direction. Such a state should break the rotational symmetry as much as possible.  Again the relevant property of the state which determines its quality as a pointer can be called its asymmetry. This example suggests that the information-theoretic point of view should be relevant for the study of asymmetry.

We will show that these two approaches to the notion of asymmetry, the constrained-dynamical and the information-theoretic, provide equivalent characterizations of asymmetry.  It follows that one can exploit the machinery of information theory for the study of asymmetry and for finding the consequences of symmetry of the dynamics.
In this paper, we will find the characterization of the G-equivalence classes of pure states using both the constrained-dynamical and the information-theoretic approaches and we will show how these two characterizations are in fact equivalent via the Fourier transform.


 In the above scenario the quantum system which is sent to Bob to transfer information about direction is called a \emph{quantum reference frame} (See \cite{BRS07} for a review of this topic). The  theory of quantum reference frames deals with the problem of  using quantum systems to transfer information, such as a direction in space, which is \emph{unspeakable}, i.e. cannot be transferred by sending a sequence of $0$s and $1$s if two agents do not have access to some shared background reference frame.  In other words, unspeakable information can only be encoded in particular degrees of freedom. For example, information about a direction in space cannot be encoded in degrees of freedom that transform trivially under rotations.

  Therefore this example suggests that the  study of asymmetry is not only useful to learn about the consequences of symmetries of dynamics but it  is also useful for the study of quantum reference frames.  The relevant property of the state which specifies how well it can act as a quantum reference frame is the asymmetry of the state. Indeed, in previous work, the asymmetry has been called the \emph{frameness} of the state \cite {GS07,GMS09}. Therefore all the results about the manipulation of reference frames and their frameness are in fact results about the asymmetry of states. In particular \cite {GS07} presents a systematic study of the manipulation of  pure state  asymmetry  for groups $U(1)$ and $Z_{2}$ and also presents some partial  results for the case of $SO(3)$. In the present paper, using a different approach based on  characterizing  the equivalence classes of asymmetries of pure states, we are able to generalize the results in \cite {GS07} significantly and to extend their scope from a few particular groups to arbitrary compact Lie groups and finite groups.

The main focus of this paper is to characterize the asymmetry of pure states. Another interesting aspect of asymmetry which has been studied previously is the problem of finding \emph{measures of asymmetry} or \emph{asymmetry monotones} \cite{Vac-Wise-Jac, Skot-Gour, Tol-Gour-Sand}. An asymmetry monotone is a function from states to real numbers which quantifies the amount of asymmetry of a state relative to a given symmetry group. This notion is mainly inspired by the notion of entanglement monotones in entanglement theory.  \footnote{Also, earlier related work has considered state-interconversion in the context of bipartite systems where two distant parties are under a U(1)-superselection rule motivated by a particle number conservation law \cite{SVC04_PRA,SVC04_PRL}.}


\subsection*{The resource theory point of view}

We can think of the study of asymmetry as \emph{a resource theory}.  Any resource theory is specified by a convex set of free states and a semi-group of free  transformations (which must  map the set of free states to itself). Any non-free state is called a \emph{resource}. The resource theory is the study of manipulations of resources under the free transformations.  As we will explain,  there are several types of questions and arguments that are relevant for all resource theories and so this point of view can help to achieve a better understanding of a specific resource theory by emphasizing its analogies with other resource theories.

A well-known example of  a resource theory is the theory of entanglement. The free transformations in this case are those which can be implemented by local operations and classical communications (LOCC).  The set of free states is the set of unentangled states. This set is closed under LOCC, i.e. an unentangled state cannot be transformed to an entangled one via LOCC \cite{Nie00}. More generally, given two quantum states one cannot necessarily transform the first one to the second with LOCC. Here the relevant properties of the states which determine whether such a transformation is possible or not are their entanglement properties.  In the case of pure bipartite states it is a well-known fact that the entanglement properties of a state are uniquely specified by its Schmidt coefficients \cite{Nie00}.  For example, Nielsen's theorem provides the necessary and sufficient condition for the existence of LOCC operations which transform one given state to another in terms of their Schmidt coefficients \cite{Nielsen-Ent}. Entangled states are also a resource in the sense that they can be used to implement tasks that are impossible by LOCC and unentangled states alone. For example, one can use entangled states for teleportation, which can be interpreted as consuming a resource (entanglement) to simulate a non-free transformation (a quantum channel) via free transformations (LOCC).

Similarly, we can think of  the study of asymmetry relative to a given representation of a group $G$ as a resource theory. In this resource theory the time evolutions which respect the symmetry (G-covariant time evolutions) are free transformations and the states which do not break the symmetry (G-invariant states) are the free states. This is a consistent choice because G-covariant time evolutions  form a semi-group under which the set of G-invariant  states is mapped to itself.  Similarly to  entanglement theory, a resource (an asymmetric state) can be used to simulate a non-free transformation (non-G-covariant time evolution)   via a free transformation (G-covariant time evolution).

In the resource theory of asymmetry,  we seek to classify different types of resources and to find the rules governing their manipulations. For every question in entanglement theory, it is useful to ask whether there is an analogous question in the resource theory of asymmetry.  In this paper, we will show that all the asymmetry properties of a pure state $\psi$ relative to the group $G$ and the unitary representation $\{U(g),g\in G\}$ are specified by its \emph{characteristic function} $\chi_{\psi}(g)\equiv\langle\psi|U(g)|\psi\rangle$. This is analogous to how all the entanglement properties of a pure bipartite state are specified by its Schmidt coefficients.


We then proceed to find the complete set of selection rules for pure states under deterministic and stochastic single-copy operations, that is, the necessary and sufficient conditions under which one pure state can be converted to another by a G-covariant operation either deterministically or nondeterministically.  These results are the analogues within the resource theory of asymmetry of, respectively, Nielsen's theorem \cite{Nielsen-Ent} and Vidal's theorem in entanglement theory.  Finally, we consider the case of catalysis of asymmetry transformations, wherein a state with asymmetry can be used to assist in the conversion but must be returned intact at the end of the protocol. We show that a finite catalyst is useless in the case of compact connected Lie groups, while in the case of a finite group, there exists for any state interconversion problem a finite catalyst that makes it possible.

\subsection*{Outline}

 We now summarize the structure of the paper.
In section \ref{sec-Prelim} we review some elementary concepts.  We also formally define G-equivalence classes of states.  Appendix \ref{app:proj} includes a short review of projective unitary representations and Appendix \ref{embeding} includes a discussion about the situations where the input and output Hilbert space of a time evolution are different. In section \ref{sec-duality},  we introduce the idea of two dual points of view to asymmetry, constrained-dynamical and information-theoretic. We also show how these two dual points of view arise naturally in the study of quantum reference frames.  In section \ref{unitary-G-equivalence}, we define  the notion of \emph{unitary G-equivalence}, another equivalence relation over states that is slightly stronger than G-equivalence. Using the  constrained-dynamical and information-theoretic perspectives, we find two different ways to characterize the unitary G-equivalence classes of states: the characteristic function and the reduction to the irreps.  Section \ref{sec-approximate} extends these considerations to the case of \emph{approximate} unitary G-equivalence, in which one state should be transformed to a state that is close to (but not necessarily exactly equal to) a second. The proofs for this section are presented in appendix \ref{app:proofofapproxGequivalence}.

In section \ref{Sec-Reduction}, we show that the two different characterizations of the unitary G-equivalence classes are in fact two different representations of the same object, the reduction of the state to the associative algebra and that these representations can be transformed one to the other via Fourier and inverse Fourier transforms. \blk
We further outline several nice mathematical properties of the characteristic function of a state, properties which make it particularly useful for the study of the asymmetry of pure states. We also show, in appendix \ref{app:distinguish}, that both the  amplitude and the phase of the characteristic function are important for specifying the asymmetry of a state, while in appendix \ref{app:charfuncs} we explain more about characteristic functions and their connection with the classical characteristic function of probability distributions.

In section \ref{G-equiv}, we present our main result, the characterization of the G-equivalence classes.  Specifically, we show that for compact Lie groups, the G-equivalence class of a state is uniquely specified by its characteristic function up to a 1-dimensional representation of the group.  In the important case of semi-simple Lie groups, we show that it is uniquely specified by the characteristic function alone.

Finally, the results on single-copy transformations are presented in the three short sections: deterministic transformations in section \ref{sec:deterministic}, state-to-ensemble transformations and stochastic transformations in section \ref{sec:stochastic}, and catalysis in section \ref{sec:catalysis}. We end with a general discussion. \blk


\section{Preliminaries}\label{sec-Prelim}

A \emph{symmetry transformation}  is a transformation which leaves the physical  objects, structures or dynamics unchanged. Group theory provides the mathematical language to describe symmetries. One can easily see that the set of symmetries of an object form a group: they are closed because if one takes a symmetry of the object, and then applies another symmetry, the total transformation will still leave the object unchanged and so is a symmetry. Furthermore, the identity transformation always leaves the object unchanged and so is a symmetry of the object. The associativity is a result of the fact that symmetries can be thought of as maps on a space, and composition of maps is associative. Finally, if a transformation leaves the object unchanged, undoing that transformation also leaves  it unchanged and so the inverse of a symmetry is also a symmetry.

In quantum theory the action of any symmetry transformation should be described by a unitary or anti-unitary acting on the Hilbert space of the system. This follows from the fact that a symmetry transformation can always be interpreted as  a change of reference frame or convention and this change should not affect the physically observable properties. In particular, it should not affect the distinguishability of states. Then, it follows from  a well-known theorem \footnote{ \textbf{Theorem:}
Let $T$ be a surjective map from a complex Hilbert space to itself such that $|\langle T\phi|T\psi\rangle|=|\langle\phi|\psi\rangle|$ for all pure states $\psi$ and $\phi$. Then $T$ has the form of  $T\psi=e^{i\theta(\psi)}V\psi$ where $\theta(\psi)$ is an arbitrary real function and $V$ is either a unitary or anti-unitary operator.}\label{Wigner}
 by Wigner  \cite{Wigner}  that any such  transformation is represented by a unitary or an anti-unitary operator on the Hilbert space of the system such that an arbitrary density operator $\rho$ is mapped by the symmetry transformation to the density operator $V\rho V^{\dag}$ for some unitary or anti-unitary operator $V$. 
In this paper we do not consider symmetry transformations, such as time-reversal, that are represented by anti-unitary operators. Therefore,  any symmetry we consider here is represented by a unitary acting on the Hilbert space of the system.

Let $G$  be a group describing  a set of symmetry transformations or a \emph{symmetry} for short. Then the action of each group element $g\in G$ should be described by a unitary $U(g)$. It follows that for consistency it should hold that for any pair of group elements $g_{1}$ and $g_{2}$ in group $G$
\begin{equation}
U(g_{2}g_{1}) \rho U^{\dag}(g_{2}g_{1})=U(g_{2})\left(U(g_{1})\rho U^{\dag}(g_{1})\right)U^{\dag}(g_{2})
\end{equation}
Since this  should hold  for any arbitrary state $\rho$ one can conclude that
\begin{equation}\label{Eq.cocycle}
U(g_{2}g_{1})=\omega(g_{2},g_{1})U(g_{2})U(g_{1})
\end{equation}
where $\omega(g_{2},g_{1})$ is a phase factor, i.e. $\left|\omega(g_{2},g_{1})\right|=1$. This means that a symmetry described by group G  should be represented  by \emph{a projective unitary representation of group G}.   The phase factor $\omega(g_{1},g_{2})$ is called the \emph{cocycle} of the representation.   We denote a specific projective unitary representation of $G$ by the set of unitaries $\{U(g), g \in G\}$ or by the map $g\rightarrow U(g)$. In the specific case where the cocycle $\omega(g_{1},g_{2})$ is constant  and equal to one, the representation is called a \emph{(non-projective) Unitary representation}.  We provide a short list of some useful properties of projective unitary representations of compact Lie groups and finite groups in  appendix \ref{app:proj}. For a helpful review of this topic we refer to chapter 2 of Giulio Chiribella's thesis \cite{Chiribella}.

We will frequently use the unitary super-operator notation to represent the action of  groups.  For any group G and any  projective unitary representation $g\rightarrow U(g)$ we define  the super-operators
\begin{equation}
\mathcal{U}_{g}(X)=U(g)X U^{\dag}(g).
\end{equation}
So under the symmetry transformation $g\in G$ the state $\rho$ will be mapped to $\mathcal{U}_{g}(\rho)$.

The representation of the fundamental symmetries of nature, such as the symmetries of space-time, are part of the specification of a physical system.
 For example, on a system with a two-dimensional Hilbert space the group of all rotations in the three-dimensional real space $\mathbb{R}^{3}$, i.e. the group SO(3), can have two different representations: the trivial representation where the action of symmetry transformations leaves all states unchanged and the non-trivial representation corresponding to the spin-half representation of SO(3). These two representations of SO(3) describe systems with different physical properties.

For most symmetries, such as the fundamental symmetries of space-time, the representation of the symmetry on a composite system is the \emph{collective representation}: if the projective unitary representations of a symmetry transformation $g\in G$ on systems with Hilbert spaces $\mathcal{H}_A$ and $\mathcal{H}_B$ are $U_A(g)$ and $U_B(g)$ respectively, then the projective unitary representation of that symmetry transformation on the Hilbert space of the composite system with Hilbert space $\mathcal{H}_A\otimes\mathcal{H}_B$ is $U_A(g)\otimes U_B(g)$. In this paper we always assume that the representation of the symmetry on the joint system  is the collective representation.


\subsection{Symmetries of states}

For any given symmetry group, there are some states which are invariant under some or all symmetry transformations in the group. For example, for any symmetry and for any representation of the symmetry, the completely mixed state is invariant under all symmetry transformations. 

\begin{definition} \label{SymSubgroup}
The \emph{symmetry subgroup} of a state $\rho$ relative to the group G, denoted $\textnormal{Sym}_G(\rho)$, is the subgroup of G under which $\rho$ is invariant,
\begin{equation}
\textnormal{Sym}_G(\rho)\equiv \{ g\in G: \mathcal{U}_g[\rho] = \rho \}.
\end{equation}
\end{definition}

If the symmetry subgroup contains only the identity element, it is said to be trivial.  In this case, it is often said that the state has \emph{no symmetries} (meaning no nontrivial symmetries). If the symmetry subgroup of a state $\rho$ is the entire group G, so that it is invariant under all symmetry transformations $g\in G$, i.e.
\begin{equation}
\forall g\in G: \mathcal{U}_{g} (\rho) = \rho,
\end{equation}
then we say that the state is \emph{G-invariant}  \footnote{ Because a symmetry transformation is defined not only by a group G but also by a representation $U$ of that group, it would be more precise to call the symmetric states ``\{G,U\}-invariant'', however, for ease of readability, we do not specify the representation explicitly. }.

\subsection{G-covariant operations} \label{sec:G-cov2}

We say that a time evolution is \emph{G-covariant} if it commutes with all symmetry transformations in the group G, that is, for any initial state and any symmetry transformation, the final state is independent of the order in which the symmetry transformation and the time evolution are applied \footnote{Again, it would be more precise to call the symmetric operators ``\{G,U\}-covariant'', however, for ease of readability, we do not specify the representation explicitly. }.  We will sometimes refer to an operation that is G-covariant as a \emph{symmetric} operation.  (It is important not to confuse symmetr\emph{y} transformations, which correspond to a particular group action, with symmetr\emph{ic} transformations, which commute with all group actions.)  We provide the rigorous form of the notion of G-covariance first for closed system evolutions and then for open system evolutions.

\begin{figure}[h!]
 \center{   \includegraphics[width=7cm]{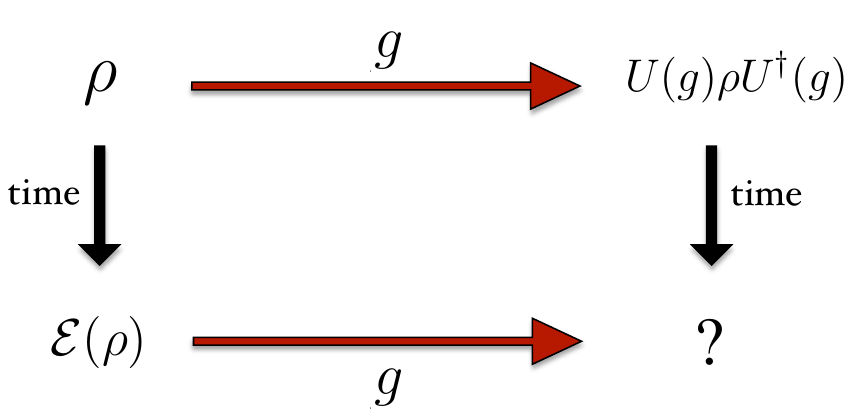}}
    \caption{\label{Fig:G-Equivalence}
   A time evolution is called G-covariant if the above transformations commute for all group elements $g\in G$.
    }
\end{figure}

Closed system dynamics are described by unitary operators over the Hilbert space.  However, noting that the global phase of a vector in Hilbert space has no physical significance, it is useful to describe the dynamics in terms of its effect on density operators (every parameter of which has physical significance).  Closed system dynamics are then described by linear maps $\mathcal{V}$ on the operator space that are of the form $\mathcal{V}[\rho]=V\rho V^{\dag}$, where $V$ is a unitary operator.  
A closed system dynamics associated with the unitary $V$ is G-covariant if
\begin{equation}
\forall g\in G,\forall \rho: V U(g)\rho U^{\dag}(g) V^{\dag} = U(g) V  \rho V^{\dag} U^{\dag}(g),
\end{equation}
or equivalently,
\begin{equation}
\forall g\in G:[\mathcal{V},\mathcal{U}_g]=0,
\end{equation}
where  $[\mathcal{V},\mathcal{U}_g]:=\mathcal{V}\circ \mathcal{U}_g-\mathcal{U}_g\circ \mathcal{V}$.  In other words, the map $\mathcal{V}$ commutes with every element of the (superoperator) representation of the group $\{ \mathcal{U}_g:g \in G\}$.  This implies that
\begin{equation}\label{G-covariant_unitary}
\forall g\in G: V U(g)= U(g) V \omega(g),
\end{equation}
where $\omega(g)$ is a phase factor that can easily be shown to be a 1-dimensional representation of the group. In the case of finite-dimensional Hilbert spaces (which is the case under consideration in this paper), we can argue that $\omega(g)=1$ if the closed system dynamics is required to be  continuous and symmetric at all times (in contrast to requiring only that the effective operation from initial to final time be symmetric) \cite{preparation}.

This argument justifies the common definition in the literature  of when a closed system dynamics respects the symmetry, namely, when
\begin{equation} \label{G-invariant_unitary}
\forall g\in G: [V, U(g)]=0.
\end{equation}
We call any unitary $V$ which satisfies this property a \emph{G-invariant unitary}
 because $\forall g\in G: U(g) V U^{\dag}(g) = V$. More generally, any operator  which commutes with the representation of group G on the Hilbert space of the system will be called G-invariant.
Clearly, if a Hamiltonian is G-invariant then all the unitaries it generates are G-invariant.  Finally, note that if $V$ is an isometry rather than a unitary, then it is said to be G-invariant if $\forall g\in G: V U_{\textnormal{in}}(g)= U_{\textnormal{out}}(g) V$, where $U_{\textnormal{in}}(g)$ and $U_{\textnormal{out}}(g)$ are the representations of the group on the input and output spaces of the isometry.
\color{black}

In general, a system might be \emph{open}, i.e., it may interact with an environment. In this case, the time evolution cannot be described by the Hamiltonian of the system alone.  Rather, to describe the time evolution we need the Hamiltonian of system and environment together. In the study of open systems we usually restrict our attention to the situations where the initial state of the system and environment are uncorrelated,
  in which case we can describe the evolution by a deterministic quantum channel  $\mathcal{E}$, that is, a
  \emph{completely positive}\footnote{Let $\mathcal{K}$  be an arbitrary Hilbert space, $\mathcal{B}(\mathcal{K})$ be the space of bounded linear operators on $\mathcal{K}$ and $\mathbb{I}_{\mathcal{B}(\mathcal{K})}$ be the identity map on  $\mathcal{B}(\mathcal{K})$. A map $\mathcal{E}$ from $\mathcal{B}(\mathcal{H}_{\textnormal{in}})$ to $\mathcal{B}(\mathcal{H}_{\textnormal{out}})$ is called \emph{completely positive} if for any Hilbert space $\mathcal{K}$, $\mathcal{E}\otimes \mathbb{I}_{\mathcal{B}(\mathcal{K})}$ is a \emph{positive} map, i.e. it maps positive operators in $\mathcal{B}(\mathcal{H}_{\textnormal{in}})\otimes  \mathcal{B}(\mathcal{K})$ to positive operators in $\mathcal{B}(\mathcal{H}_{\textnormal{out}})\otimes  \mathcal{B}(\mathcal{K})$.}, trace-preserving, linear map from $\mathcal{B}(\mathcal{H}_{\textnormal{in}})$ to $\mathcal{B}(\mathcal{H}_{\textnormal{out}})$ where $\mathcal{H}_{\textnormal{in}}$  and $\mathcal{H}_{\textnormal{out}}$ are the input and output Hilbert spaces and $\mathcal{B}(\mathcal{H})$ are the bounded operators on $\mathcal{H}$.
  After a time evolution described by quantum channel $\mathcal{E}$, the initial state $\rho$ evolves to the final state $\mathcal{E}(\rho)$.
 Note that a general quantum channel may have input and output spaces that are distinct. This possibility is useful for describing transformations wherein the system of interest may grow (by incorporating into its definition parts of the environment) or shrink (by having some of its parts incorporated into the environment).

We now state the conditions for a general quantum operation (which may represent open or closed system dynamics) to be G-covariant.
\begin{definition}[G-covariant operation] The quantum operation $\mathcal{E}$ is said to be \emph{G-covariant} if
\begin{equation}
\forall g\in G:\  \mathcal{E}\left(U_{\textnormal{in}}(g)(\cdot)U_{\textnormal{in}}^\dag(g)\right)=U_{\textnormal{out}}(g)\mathcal{E}\left(\cdot\right)U_{\textnormal{out}}^\dag(g),
\end{equation}
where $\{U_{\textnormal{in}}(g): g\in G \}$ and $\{U_{\textnormal{out}}(g):g\in G\}$ are the representations of $G$ on the input and output Hilbert spaces of $\mathcal{E}$.
\end{definition}

If the input and output spaces are equivalent
then the condition of G-covariance can be expressed as
\begin{equation}\label{covariance}
\forall g\in G:\  \mathcal{E}\left(U(g)(\cdot)U^\dag(g)\right)=U(g)\mathcal{E}\left(\cdot\right)U^\dag(g),
\end{equation}
or equivalently,
\begin{equation}
\forall g\in G:[\mathcal{E},\mathcal{U}_g]=0,
\end{equation}
where $\mathcal{U}_g[\cdot]= U(g)(\cdot)U^{\dag}(g)$.

 As we demonstrate in appendix \ref{embeding}, any G-covariant operation for which the input and output Hilbert spaces are different can always be modeled by one wherein the input and output Hilbert spaces are the same. The reason is that the input and output Hilbert spaces can always be taken to be two different sectors of a single larger Hilbert space, $\mathcal{H}_{\textnormal{in}}\bigoplus \mathcal{H}_{\textnormal{out}}$, and any operation from $\mathcal{B}(\mathcal{H}_{\textnormal{in}})$ to $\mathcal{B}(\mathcal{H}_{\textnormal{out}})$ that is G-covariant relative to the representations $\{U_{\textnormal{in}}(g)\}$ and $\{U_{\textnormal{out}}(g)\}$
can always be extended to an operation on $\mathcal{B}(\mathcal{H}_{\textnormal{in}}\bigoplus \mathcal{H}_{\textnormal{out}})$ that is G-covariant relative to the representation $\{U_{\textnormal{in}}(g)\bigoplus U_{\textnormal{out}}(g)\}$.
Similarly, any G-invariant isometry (a reversible operation where the input and output Hilbert spaces may differ) can always be modelled by a G-invariant unitary (where the input and output Hilbert spaces are the same).  Again, this is shown in appendix \ref{embeding}. It follows that without loss of generality, we can restrict our attention in the rest of this paper to G-covariant operations where the input and output spaces are the same.


Clearly, $G$-covariant quantum operations include those induced by $G$-invariant unitaries, that is, operations of the form $\mathcal{V}(\cdot)=V(\cdot)V^{\dag}$ where $\forall g\in G:[V,U(g)]=0$.
As another example, consider a channel of the form
\begin{equation}
\mathcal{K}\equiv\int_{K}dk\ \mathcal{U}_k,
\end{equation}
where $K$ is a subgroup of $G$ and $dk$ is the uniform measure over $K.$  We refer to this as the \emph{uniform twirling over }$K$~\footnote{Note that we can implement the time evolution described by the channel $\mathcal{K}$ by choosing one of the unitaries from the set  $\{U(k), k\in K\}$ uniformly at random and applying it to the system.
}.  The uniform twirling over any normal subgroup of G is a G-covariant operation.
\ First, recall that if $K$ is a normal subgroup of $G$ then
$\forall g\in G:gKg^{-1}=K,$
where $gKg^{-1}\equiv\left\{  gkg^{-1}:k\in K\right\}  .$  It follows that
\begin{equation}
\forall g\in G:\mathcal{U}_g\circ\mathcal{K}\circ\ \mathcal{U}_{g^{-1}}=\int
_{K}dk\ \mathcal{U}_{gkg^{-1}}=\mathcal{K},
\end{equation}
and consequently that $\mathcal{K}$ is G-covariant. In particular any group is the normal subgroup of itself, therefore uniform twirling over any  group $G$ is a G-covariant channel.


Furthermore, if we couple the object system to an environment using a Hamiltonian which has the symmetry $G$ and if the environment is initially uncorrelated with the system and prepared in a state that is G-invariant, and finally some \emph{proper} subsystem is discarded, then the total effect of this time evolution is described by a G-covariant quantum operation. (Intuitively this is clear, because there is nothing in such a dynamics that can break the symmetry.) Here by \emph{proper} subsystem we mean a subsystem  which is closed under the action of the symmetry transformations, i.e., under this action any vector in that subsystem is mapped to a vector in the same subsystem.

As it turns out, \emph{every} $G$-covariant quantum operation can in fact be realized in this way, i.e. by first coupling the system to an uncorrelated environment in a G-invariant state via a G-invariant unitary and secondly discarding a proper subsystem of the total system. This is sometimes called  the Stinespring dilation theorem for G-covariant channels  and was first proved in \cite{Werner}.~\footnote{A different proof of this is provided in \cite{thesis:Marvian}. } This result provides an operational prescription for realizing every such operation.

In the theory of asymmetry we  study the consequences of the fact that a (possibly open) dynamics  has a symmetry.  In particular, we are interested to know,  for a given initial state of a G-covariant dynamics, which kind of constraints one can put on the possible final states based on the symmetries of dynamics.
Equivalently,  we are interested to know, for a given pair of states $\rho$ and $\sigma$, whether  there exists a G-covariant dynamics which transforms $\rho$ to $\sigma$ or not.  We use the notation $\rho\xrightarrow{\textnormal{G-cov}} \sigma$ to denote that state $\rho$ can be transformed to state $\sigma$ under a G-covariant time evolution.

For instance, a simple consequence of the symmetry of dynamics is that every symmetry of the initial state is a symmetry of the final state, i.e.,
\begin{proposition} \label{prop_one}
If $\rho$ transforms to $\sigma$ by a G-covariant quantum operation ( $\rho\xrightarrow{\textnormal{G-cov}} \sigma$), then  $\textnormal{Sym}_G(\rho)\subseteq \textnormal{Sym}_G(\sigma)$.
\end{proposition}
\begin{proof}
If $g_{s}\in G$ is a symmetry of $\rho$ then $\mathcal{U}_{g_{s}}(\rho)=\rho $. Since the operation $\mathcal{E}$ taking $\rho$ to $\sigma$ is G-covariant,
it follows that
$$\mathcal{E}(\rho)=\mathcal{E}\circ \mathcal{U}_{g_{s}}(\rho)= \mathcal{U}_{g_{s}}\circ\mathcal{E}(\rho)$$
So $\mathcal{U}_{g_{s}} (\sigma)=\sigma$.	
\end{proof}

In particular, therefore, one cannot generate an asymmetric state starting from a symmetric one.  This proposition highlights a simple example of restrictions one can put on the final states of a possibly open system dynamics based on the initial state of the system and symmetry of dynamics.  For instance, it implies that under rotationally-covariant time evolutions, a spin pointing along $\hat{z}$ cannot evolve to one pointing along $\hat{x}$ because the first state is invariant under the group of rotations around $\hat{z}$ while the second one is not. This result can be understood as a cognate of Curie's principle, which states that symmetric causes cannot have asymmetric effects  \cite{Curie}. Also,  note that  this proposition  suggests a simple characterization of the asymmetry of states relative to a group $G$ by characterizing the largest subgroup of $G$ which leaves each state invariant. Indeed, this simple characterization is very useful, for example, in condensed matter theory. However,  finding a more fine-grained characterization of asymmetry of states can also be useful, for example, to study the consequences of symmetry of an open system dynamics.

On the other hand, for any arbitrary pair of G-invariant states $\rho$ and $\sigma$ there always exist G-covariant channels which transform one to the other. A trivial instance of these G-covariant channels, is the one which discards the input state and generates the G-invariant state $\sigma$ as the output, i.e. the channel described by
\begin{equation}
\mathcal{E}_{\sigma}(X)=\textnormal{tr}(X) \sigma
\end{equation}

Finding the necessary and sufficient condition to determine  for any given pair of states $\rho$ and $\sigma$ whether   $\rho\xrightarrow{G-\textnormal{cov}}\sigma$ or not,  turns out to be a hard problem and is still open.  However, in this paper, we will answer this question for the special case where both $\rho$ and $\sigma$ are pure states. In the rest of this section we present two physical examples of channels which are covariant with respect to the group U(1), the group formed by all phases $\{e^{i\theta}:\theta\in(0,2\pi]\}$.

\subsection{Example: U(1)-covariant channels} \label{sec-q-opt-ex}

For concreteness, it is worth examining a specific example of symmetric operations, namely, those that are covariant under a unitary representation of the U(1) group.
Here, we present two different physical scenarios in which a restriction to U(1)-covariant channels is natural.

\subsubsection{Axially symmetric channels}
U(1)-covariant quantum operations are relevant for describing a dynamics which has rotational  symmetry around some axis, or \emph{axially symmetric} dynamics. The set of all rotations around a fixed axis forms the group called SO(2) which is isomorphic to the group U(1). So the unitary representation of the rotations around a fixed axis forms a representation of U(1), e.g. if $L_{z}$ is the operator of angular momentum in the $z$ direction then
 $$e^{i\theta}\rightarrow e^{i\theta L_{z}}$$
 is a representation of the group $U(1)$.
 In general the eigenvalues of $L_{z}$ are  degenerate. But to simplify the notation here we assume $L_{z}$ has no degeneracy. So $\{|m\rangle: m\in \{-j,-j+1,\cdots  ,j\} \}$, the eigenbasis of $L_z$,  is a basis for the Hilbert space of the system, where $j$ is the angular momentum of the system and so is either half integer or integer and where $L_{z}|m\rangle=m|m\rangle$ (taking $\hbar=1$). Note that in the case of half-integer spins the representation  $e^{i\theta}\rightarrow e^{i\theta L_{z}}$ is a projective representation, i.e. the cocycle of the representation is non-trivial.

First, we consider the symmetries of a few different states. The state $(|0\rangle + |1\rangle)/\sqrt{2}$ has no symmetries, while the state $(|0\rangle + |2\rangle)/\sqrt{2}$ has a nontrivial symmetry subgroup because it is invariant under a $\pi$ phase shift.  Meanwhile, all the elements of the basis $\{|m\rangle: m\in \{-j,-j+1,\cdots  ,j\} \}$ are U(1)-invariant states. The set of all states (pure and mixed) that are U(1)-invariant are those which commute with all elements of the set $\{\exp{(i \theta {L_{z}})}:\theta\in (0,2\pi]\}$ and so commute with ${L_{z}}$ and are therefore diagonal in the $\{|m\rangle: m\in \{-j,-j+1,\cdots  j\} \}$ basis.

Next we consider symmetric operations.  First note that the \textnormal{U}(1)-invariant unitaries are those that are diagonal in the $\{|m\rangle: m\in \{-j,-j+1,\cdots  j\} \}$ basis and are therefore of the form
\begin{equation}
V_{\textnormal{U(1)-inv}}=\sum_{m=-j}^{j} e^{i\beta_{m}} |m\rangle\langle m|
\end{equation}
These unitaries all commute with each other.  (Note, however, that if there is multiplicity in the representations, then the U(1)-invariant unitaries have a more complicated structure and do not necessarily commute with each other.) \blk

Now one can easily see that using U(1)-invariant unitaries we cannot transform one arbitrary state to another. For example, we cannot transform $|0\rangle$ to $(|0\rangle+|1\rangle)/\sqrt{2}$: The first state is a symmetric state while the second has some asymmetry. Similarly we can easily see that $(|0\rangle+|1\rangle)/\sqrt{2}$ cannot be transformed to $(|2\rangle+|3\rangle)/\sqrt{2}$ using U(1)-invariant unitaries. However, this transformation \emph{is} possible using a U(1)-covariant channel. Consider the quantum operation $\mathcal{E}$ described by the following Kraus operators:
$$K_{0}=\sum_{m=-j}^{j-1} |m+1\rangle\langle m|,\ \  \textnormal{and}\ \ K_{1}= |-j\rangle\langle j|$$
where  $K_{0}^{\dag} K_{0}+K_{1}^{\dag}K_{1}=I$. One can easily check that this quantum operation is covariant under rotations around $\hat{z}$, i.e.
\begin{equation}
\forall\theta\in (0,2\pi]:\ \mathcal{E}(e^{i \theta L_{z}} \rho e^{-i \theta L_{z}}  )=e^{i \theta L_{z}}  \mathcal{E}(\rho ) e^{-i \theta L_{z}}.
\end{equation}
Furthermore, it maps the state $ (|m-1\rangle+ |m\rangle)/\sqrt{2}$ to $(|m\rangle+ |m+1\rangle)/\sqrt{2}$ for all $m<j$.
So, although the transformation is not possible via U(1)-invariant unitaries, it can be done by U(1)-covariant quantum operations. Similarly we can show that there is a U(1)-covariant quantum operation which transforms $(|m\rangle+ |m+1\rangle)/\sqrt{2}$ to $ (|m-1\rangle+ |m\rangle)/\sqrt{2}$.


\subsubsection{Phase-covariant channels in quantum optics}\label{q-optics-example}

Another physical example of $U(1)$-covariant quantum operations comes from quantum optics (for more discussion see \cite{BRS07} ).
Consider a harmonic oscillator whose Hilbert space is spanned by the orthonormal basis $\{|n,\alpha\rangle:  n\in\mathbb{N}\}$ with the number operator ${N}$ such that  ${N}|n,\alpha\rangle=n |n,\alpha\rangle$ where $n$ is  a nonnegative integer and $\alpha$ labels possible degeneracies. Then the operator which shifts this oscillator in its cycle by phase $\theta$ is $\exp{(i \theta {N})}$. For example, this operator transforms the coherent state $|\gamma\rangle$ to $|e^{{i\theta}}\gamma\rangle$.

Now a quantum operation $\mathcal{E}$ is phase-covariant if
 \begin{equation}
\forall\theta\in (0,2\pi]:\ \mathcal{E}(e^{i \theta {N}} \rho e^{-i \theta {N}}  )=e^{i \theta {N}}  \mathcal{E}(\rho ) e^{-i \theta {N}}.
\end{equation}

For a particular physical scenario, there may be additional constraints on the accessible states and unitaries beyond those that are implied by the symmetry.  For instance, here in this example, unlike the previous example, there is no invariant state which under the action of the symmetry group transforms as $e^{i N\theta}|\psi\rangle=e^{-i\theta}|\psi\rangle$; all eigenvalues of the number operator are non-negative. This is a restriction relative to what occurs for our first example where to realize  a particular axially symmetric operation an experimenter can couple the system to an ancilla in state $\{|m\rangle\}$ for arbitrary  positive or negative $m$.

However, it turns out that a restriction of the accessible irreps of U(1) to the nonnegative does not have any impact on the set of operations one can implement -- all U(1)-covariant operations are still physically accessible \cite{thesis:Marvian}. In other words, any phase-invariant quantum operation can be realized by coupling the system to another ancillary system which is initially in $|n\rangle$ for some non-negative $n$ and the coupling can be chosen to be a phase-invariant unitary.\footnote{This follows from the constructive proof of Stinespring dilations of G-covariant channels presented in \cite{thesis:Marvian}.}  For the rest of this paper, we will assume that all G-covariant operations are physically accessible (including in the quantum optics examples).

\section{Asymmetry of quantum states}\label{sec-duality}

The \emph{asymmetry properties} of a state relative to some symmetry group specify how and to what extent the given symmetry is broken by the state.   Characterizing these is found to be surprisingly useful for addressing a very common problem: to determine what follows from a system's dynamics (possibly open) having that symmetry.  In this section we formally define the notion of asymmetry of a state and demonstrate  that the asymmetry properties of a state can be understood in terms of information-theoretic concepts.



The first step in characterizing asymmetry is to specify when two states have the same asymmetry.  We stipulate that this is the case when
the pair of states can be \emph{reversibly interconverted} one to the other by symmetric operations.  This defines an equivalence relation among states.
\begin{definition}[G-equivalence of states]\label{Gequiv}
Two states, $\rho$ and $\sigma$, are said to be \emph{$G$-equivalent} if and only if they are reversibly interconvertible by $G$-covariant operations, i.e.,
there exists a quantum operation $\mathcal{E}$ such that
\begin{equation}
\forall g\in G:[\mathcal{E},\mathcal{U}_g]=0 , \;\textnormal{ and }\; \mathcal{E}[\rho]=\sigma,
\end{equation}
and there exists a quantum operation $\mathcal{F}$ such that
\begin{equation}
\forall g\in G:[\mathcal{F}, \mathcal{U}_g]=0, \;\textnormal{ and }\; \mathcal{F}[\sigma]=\rho.
\end{equation}
\end{definition}


\noindent(Using the notation we introduced in section \ref{sec:G-cov2}, $\rho$ and $\sigma$ are G-equivalent iff $\rho\xrightarrow{\textnormal{G-cov}} \sigma$ and $\sigma\xrightarrow{\textnormal{G-cov}} \rho$.)

\begin{figure}[h!]
 \center{   \includegraphics[width=8cm]{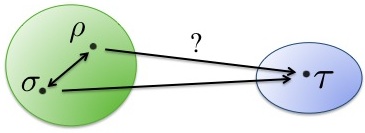}}
    \caption{\label{Fig:G-Equivalence}
   A depiction of two G-equivalence classes in the space of all states. Because both $\rho\xrightarrow{\textnormal{G-cov}} \sigma$ and $\sigma\xrightarrow{\textnormal{G-cov}} \rho$ are possible, $\rho$ and $\sigma$ are  in the same class. It follows that if $\rho\xrightarrow{\textnormal{G-cov}} \tau$ then $\sigma\xrightarrow{\textnormal{G-cov}} \tau$.
    }
\end{figure}

 A complete specification of the G-asymmetry properties of a state is achieved by specifying its G-equivalence class.
 So, for example specifying the G-equivalence class of a state should include a specification of the state's symmetries (indeed, this can be considered to be a condition that must be satisfied by any proposed specification of the asymmetry properties).  To see this first note that, as it is highlighted in proposition \ref{prop_one}, if $\rho$ can be transformed to  $\sigma$ by a G-covariant quantum operation ( $\rho\xrightarrow{\textnormal{G-cov}} \sigma$), then  $\textnormal{Sym}_G(\rho)\subseteq \textnormal{Sym}_G(\sigma)$ where $\textnormal{Sym}_G(\rho)$ is the subgroup of $G$ which leaves $\rho$ invariant (see Definition~\ref{SymSubgroup}).  So if $\rho$ and $\sigma$ are G-equivalent, i.e., $\rho\xrightarrow{\textnormal{G-cov}} \sigma$ and $\sigma\xrightarrow{\textnormal{G-cov}} \rho$, then $\textnormal{Sym}_G(\rho)= \textnormal{Sym}_G(\sigma)$.

As another example, if we want to know   whether there exists a one-way (deterministic or stochastic) symmetric transformation from one given state to another, all we need to know is the G-equivalence class of the two states; if there exists a symmetric transformation from one member of class I to one member of class II, then there exists a symmetric transformation from every member of class I to every member of class II. So to answer the question of whether a given state can evolve to another state under a G-covariant dynamics, the only properties of the two states which are relevant are their G-asymmetry properties.



The above definition of asymmetry properties  is based on the intuition that asymmetry is something which cannot be generated by symmetric time evolutions.  We call this the \emph{constrained-dynamical} perspective.

In the constrained-dynamical point of view, we characterized the asymmetry properties of a state as those features that are required to determine whether any pair of states are reversibly interconvertible by symmetric operations.

It seems natural in this point of view, to use dynamical concepts to describe and study asymmetry. For example if the symmetry group  under consideration is the rotation group, then we may use angular momentum to describe asymmetry: we know that if the expectation value of any component of the angular momentum is nonzero then the state  necessarily breaks the rotational symmetry and so is asymmetric. Moreover according to Noether's theorem, in an isotropic closed time evolution
 every component of the angular momentum is conserved.  We can generalize this result to symmetric reversible transformations on open systems using a Carnot style of argument --- in a reversible transformation the environment cannot be a source of angular momentum and therefore if a transformation can be achieved reversibly on the system alone, then it must conserve all components of angular momentum (on pain of allowing a cycle that generates arbitrary amounts of angular momentum).   It follows that the expectation value of angular momentum is a function of the $G$-equivalence class, i.e. it is the same for all states in the same G-equivalence class.

So clearly, dynamical concepts provide a useful framework for describing asymmetry. In the next section we show that information-theoretic concepts are also useful for the study of asymmetry.

\subsection{Information-theoretic point of view to asymmetry}

In this section we introduce another perspective to the notion of asymmetry of states which we call the  \emph{information-theoretic} perspective. \footnote{ Recently, a similar information-theoretic argument has been used in \cite{Vac2012} to study the duality between the particle and wave natures of quantum systems from the point of view of symmetry and asymmetry.} Recall that a quantum state breaks a symmetry, say rotational symmetry, if for some non-trivial rotations, the rotated version of the state is not the same as the state itself, i.e. they are distinguishable.   In this case, the ensemble of states corresponding to the orbit of the state under rotations can act as an encoding when the message to be encoded is an element of the rotation group. This suggests that   information-theoretic concepts are also useful for the study of asymmetry.



  Consider a set of communication protocols in which one chooses a message $g\in G$ according to a measure over the group and then sends the state $\mathcal{U}_g[\rho]$ where $\rho$ is some fixed state. The goal of the sender is to inform the receiver about the specific chosen group element. We claim that the asymmetry properties of a state $\rho$ can be defined as those that determine  the effectiveness of using the signal states $\{ \mathcal{U}_g[\rho] : g\in G \}$ to communicate a message $g\in G$. To get an intuition for this, note that if $\rho$ is invariant under the effect of some specific group element $h$ then the state used for encoding $h$ would be the same as the state used for encoding the identity element $e$, ($\mathcal{U}(h)[\rho]=\mathcal{U}(e)[\rho]=\rho$), such that the message $h$ cannot be distinguished from $e$.
  In the extreme case where $\rho$ is invariant under all group elements this encoding does not transfer any information. 

 So from this point of view, the asymmetry properties of $\rho$ can be inferred from the information-theoretic properties of the encoding $\{\mathcal{U}_g[\rho]: g\in G\}$.  To compare the asymmetry properties of two arbitrary states  $\rho$ and $\sigma$, we have to compare the information content of two different encodings: $\{\mathcal{U}_g[\rho]: g\in G\}$ (encoding I) and $\{\mathcal{U}_g[\sigma]: g\in G\}$ (encoding II).   If each state $\mathcal{U}_g[\rho]$ can be converted to   $\mathcal{U}_g[\sigma]$ for all $g\in G$, then encoding I has as much or more information about $g$ than encoding II.  If the opposite conversion can also be made, then the two encodings have precisely the same information about $g$.  Consequently, in an information-theoretic characterization of the asymmetry properties, it is the reversible interconvertibility of the sets (defined by the two states) that defines equivalence of their asymmetry properties.

As it turns out, our two different approaches lead to the same definition of asymmetry properties, as the following lemmas imply.

\begin{lemma}\label{lemma1}
The following statements are equivalent:\\
\textbf{A})\ There exists a G-covariant quantum operation $\mathcal{E}_{\textnormal{G-cov}}$ [as defined in Eq.~(\ref{covariance})] which maps $\rho$ to $\sigma$,
i.e., $\mathcal{E}_{\textnormal{G-cov}}(\rho)=\sigma$\\
\textbf{B})\  There exists a  quantum operation $\mathcal{E}$ which maps $\mathcal{U}_g[\rho]$ to $\mathcal{U}_g[\sigma]$ for all $g\in \textnormal{G}$,
i.e.,
\begin{equation} \label{Bpoint}
\forall g\in G:\ \ \   \mathcal{E}(\mathcal{U}_g[\rho])=\mathcal{U}_g[\sigma].
\end{equation}
\end{lemma}
For pure states, we have
\begin{lemma}\label{unitary}
The following statements are equivalent:\\
\textbf{A})\ There exists a G-invariant unitary  $V_{\textnormal{\textnormal{G-inv}}}$ (i.e. $\forall g\in G:\ [V_{\textnormal{G-inv}},U(g)]=0$) which maps $\ket{\psi}$ to $\ket{\phi}$, i.e. $V_{\textnormal{G-inv}}|{\psi}\rangle=|{\phi}\rangle$. \\
\textbf{B})\ There exists a unitary operation $V$ which  maps $U(g) \ket{\psi}$ to $U(g)\ket{\phi}$ for all $g\in G$, i.e.,
\begin{equation}
\forall g\in G:\ \ \   V U(g)|{\psi}\rangle=U(g)|{\phi}\rangle.
\end{equation}
\end{lemma}

 Note that in both of these lemmas, the condition $\textbf{A}$ concerns whether it is possible to transform a single state to another under a limited type of dynamics. On the other hand, in the \textbf{B} condition, there is no restriction on the dynamics, but now we are asking whether one can transform a \emph{set} of states to another set such that each state in the former set is mapped to its corresponding state in the latter set under this dynamics.

Adopting the latter perspective enables us to use the machinery of quantum information theory to study asymmetry and, via the lemmas, the consequences of symmetric dynamics.  This technique has many other applications in the study of asymmetry. For instance, the information-theoretic approach is used in \cite{thesis:Marvian} to quantify the amount of asymmetry of states. In this paper we will find the characterization of the G-equivalence classes of pure states using both the constrained-dynamical and the information-theoretic approaches and we will show how these two characterizations are in fact equivalent via the Fourier transform. Also in the next section we explain how these two different perspectives on asymmetry naturally arise in the study of uncorrelated reference frames.  First however,  we present the proofs of the lemmas.

\begin{proof} (Lemma \ref{lemma1})
 \textbf{A} can be seen to imply \textbf{B} by taking $\mathcal{E}=\mathcal{E}_{\textnormal{G-cov}}$. To show the reverse, note that \textbf{B} implies the existence of a quantum operation $\mathcal{E}$ which satisfies Eq.~(\ref{Bpoint}). Now we can define
 \begin{equation} \label{twirling}
 \mathcal{E}'\equiv \int dg\ \mathcal{U}^\dag_{g}\circ \mathcal{E} \circ\mathcal{U}_g
\end{equation}
One can then easily check that $\mathcal{E}'$ is a G-covariant operation and that $\mathcal{E}'(\rho) = \int dg\ \mathcal{U}^\dag_{g}\circ \mathcal{E} \circ\mathcal{U}_g(\rho)= \int dg\ \mathcal{U}^\dag_{g}\circ \mathcal{U}_g(\sigma)=\sigma$, such that we can choose $\mathcal{E}_{\textnormal{G-cov}}=\mathcal{E}'.$
So \textbf{B} also implies \textbf{A}.
\end{proof}

\begin{proof} (Lemma \ref{unitary})
\textbf{A} can be seen to imply \textbf{B} by taking $V=V_{\textnormal{G-inv}}$. In the following we prove that \textbf{B} also implies \textbf{A}. Assume there exists a unitary $V$ such that $\forall g \in G$,
\begin{equation}
V U(g) \ket{\psi}=U(g)\ket{\phi}.
\end{equation}
First note that this implies $\ket{\phi}=V \ket{\psi}$.  Furthermore it implies that for all $g,h\in G$ we have
\begin{align*}
 V U(g) U(h)|{\psi}\rangle&=\omega(g,h)V U(gh)|{\psi}\rangle\\&=\omega(g,h) U(gh) |{\phi}\rangle\\ &= U(g)U(h) |{\phi}\rangle\\ &= U(g) V U(h)|{\psi}\rangle
\end{align*}
where we have used the fact that $g\rightarrow U(g)$  is a projective representation of G and so $U(g)U(h)=\omega(g,h)U(gh)$ for a phase $\omega(g,h)$. Now suppose $\Pi$ is the projector to the subspace spanned by all the vectors $\{U(h)\ket{\psi}, \forall h\in G\}$. Then the above equation implies that
\begin{equation}
\forall g\in G:\ V U(g)\Pi= U(g)V \Pi.
\end{equation}
Now by  definition of the projector $\Pi$ it is clear that  it commutes with all $\{U(g):g\in G\}$. So the above equation implies
\begin{equation}
\forall g \in G: [V\Pi,U(g)]=0.
\end{equation}
The operator $V\Pi$ unitarily maps a subspace of the Hilbert space to another subspace and it commutes with all $\{U(g)\}$. Using lemma \ref{extension} we conclude that this G-invaraint isometry can always be extended to a G-invariant  unitary $V_{\textnormal{G-inv}}$ such that $V_{\textnormal{G-inv}}\Pi=V\Pi$ and therefore
\begin{equation}
V_{\textnormal{G-inv}} U(g)|\psi\rangle= V\Pi U(g)|\psi\rangle=U(g)|\phi\rangle.
\end{equation}
\end{proof}

\subsection{Interpreting the two points of view in terms of uncorrelated reference frames} \label{reference frames}

Interestingly these two points of view to asymmetry naturally arise in the study of a communication scenario when the two distant parties lack a shared reference frame  for some degree of freedom.

Specifically, consider a degree of freedom that transforms according to the group $G$.  Passive transformations of the reference frame for this degree of freedom will then also be described by the group $G$, as will the relative orientation of any two such frames.  Consider two parties, Alice and Bob, that each have a local reference frame but where these are related by a group element $g\in G$ that is unknown to either of them.  For instance, they might each have a local Cartesian frame, but do not know their relative orientation. (See Ref. \cite{BRS07} for a discussion.)

Now consider the following state interconversion task. Alice prepares a system in the state $\rho$ relative to her local reference frame and sends it, along with a classical description of $\rho$, to Bob.  She also sends him a classical description of a state $\sigma$, and asks him to try and implement an operation that leaves the system in the state $\sigma$ relative to her local frame.    In effect, Alice is asking Bob to transform $\rho$ to $\sigma$ but without the benefit of having a sample of her local reference frame.  For instance, she may ask him to transform a spin aligned with her $\hat{z}$-axis to one that is aligned with her $\hat{y}$-axis.  We consider how the task is described relative to each of their local frames.




 \textbf{Description relative to Alice's frame}.   In this case, the initial and final states, $\rho$ and $\sigma,$ are described relative to Alice's frame. If the operation that Bob implements is described as $\mathcal{E}$ relative to his frame, then it would be described as  $\mathcal{U}^{\dag}(g)\circ\mathcal{E}\circ\mathcal{U}_g$ relative to Alice's frame by someone who knew which group element $g$ connected their frames. However,
since $g$ is unknown to Alice and Bob, they describe the operation relative to Alice's frame by the uniform mixture of such operations, i.e., by $\int dg\mathcal{U}_g\circ$ $\mathcal{E}\circ\mathcal{U}_{g}^{\dag}$. It is straightforward to check that this quantum operation is G-covariant.  So all the operations that Bob can implement are described relative to Alice's frame as G-covariant operations.  From this perspective, the interconversion can be achieved only if $\rho$ can be mapped to $\sigma$ by a G-covariant quantum operation.

 \textbf{Description relative to Bob's frame}.  The initial state is described as $\mathcal{U}_g[\rho]$ relative to Bob's frame. \ Bob must implement an operation that transforms this to a state which is described as $\mathcal{U}_g[\sigma]$ relative to his frame. But the group
element $g$ that connects Alice's to Bob's frames is unknown, therefore the transformation is required to succeed regardless of $g$.  Bob can implement any operation relative to his own frame and so the set of operations to which he has access is unrestricted.  The question, therefore, is whether there exists an operation $\mathcal{E}$ such that $\forall g\in G:$ $\mathcal{E}\left[  \mathcal{U}_g[\rho]\right]=\mathcal{U}_g[\sigma].$   In other words, from this perspective the interconversion task can be achieved only if every element of the set $\left\{  \mathcal{U}_g[\rho]: g\in G\right\}  $ can be mapped to the corresponding element of $\left\{  \mathcal{U}_{g}[\sigma]: g\in G\right\} $ by the same quantum operation.


We see therefore that the constrained-dynamical and information-theoretic points of view to the manipulation of asymmetry arise naturally as Alice's and Bob's points of view respectively.  They constitute the descriptions of a single interconversion task relative to two different reference frames.



\section{Unitary G-equivalence } \label{unitary-G-equivalence}

In the previous section we defined the notion of G-equivalence classes of states and we argued that the G-equivalence class of a state specifies all its asymmetry properties.

It is useful to introduce another equivalence relation over states that is slightly stronger than G-equivalence. Let $g\rightarrow U(g)$ be the projective unitary representation of the symmetry described by group G  on the Hilbert space of a system. Then
\begin{definition}[Unitary G-equivalence] \label{unitary G-equivalence I}
Two pure states, $\psi$ and $\phi$, are called \emph{unitarily G-equivalent} if they are interconvertible by a G-invariant unitary, that is, if there exists a unitary $V_{\textnormal{G-inv}}$ such that
$\forall g\in G: [V_{\textnormal{G-inv}},U(g)]=0$
and
\begin{equation}\label{Unitary-G-equivalence}
V_{\textnormal{G-inv}} |\psi\rangle= |\phi \rangle
\end{equation}
\end{definition}

Recall the two alternative points of view to the notion of asymmetry introduced in the previous section, i.e.  the constrained-dynamical point of view and the information-theoretic point of view.
This definition is based on the constrained-dynamical point of view. Alternatively we can define this concept in the information-theoretic point of view in terms of the unitary interconvertibility of the orbits defined by the two states. The equivalence of these two definitions follows trivially from lemma \ref{unitary}.

As we will see later, it turns out that for connected compact Lie groups it is a small step from characterizing unitary G-equivalence to characterizing general G-equivalence. In particular  in section \ref{G-equiv}, we will show that  for semi-simple connected compact Lie groups the unitary $G$-equivalence classes are the same as the $G$-equivalence classes.

\subsection{The constrained-dynamical characterization: equality of the reductions onto irreps} \label{reductions onto irreps}

We here find a characterization of the unitary G-equivalence classes within the constrained-dynamical perspective.  We begin by determining the most general form of a G-invariant unitary.

Suppose $\{U(g):g\in G \}$ is a projective unitary representation of a finite or compact Lie group $G$ on the Hilbert space $\mathcal{H}$.  We can always decompose this representation to a discrete set of finite-dimensional irreducible projective unitary representations (irreps).  This suggests the following decomposition of the Hilbert space  \cite{BRS07},
\begin{equation}
\mathcal{H}=\bigoplus_{\mu} \mathcal{M}_{\mu}\otimes  \mathcal{N}_{\mu},
\end{equation}
where $\mu$ labels the irreps and $\mathcal{N}_{\mu}$ is the subsystem associated to the copies of representation $\mu$ (the dimension of $ \mathcal{N}_{\mu}$ is equal to the multiplicity of the irrep $\mu$ in this representation). Then  $U(g)$ can be written as
\begin{equation}  \label{decomposition}
U(g) =\bigoplus_\mu U_\mu(g) \otimes \mathbb{I}_{\mathcal{N}_\mu}
\end{equation}
where $U_\mu(g)$ acts on $\mathcal{M}_{\mu}$ irreducibly and where $\mathbb{I}_{\mathcal{N}_\mu}$ is the identity operator on the multiplicity subsystem $\mathcal{N}_{\mu}$. We denote by $\Pi_{\mu}$ the projection operator onto the subspace  $\mathcal{M}_{\mu}\otimes  \mathcal{N}_{\mu}$, the subspace associated to the irrep $\mu$.

Now we are ready to characterize the unitary G-equivalence classes:

\begin{theorem}\label{Thm-reductions}
Two pure states $|\psi\rangle$ and $|\phi\rangle$ are unitarily G-equivalent if and only if
\begin{equation} \label{reduction}
\forall \mu:\  \textnormal{tr}%
_{\mathcal{N}_{\mu}}({\Pi}_{\mu}\ |\psi\rangle\langle\psi|\ {\Pi}_{\mu})=\textnormal{tr}_{\mathcal{N}_{\mu}}({\Pi}_{\mu}\ |\phi\rangle\langle\phi|{\Pi
}_{\mu})
\end{equation}
\end{theorem}

\begin{proof}
First, we find a simple characterization of G-invariant unitaries.
It is shown in appendix~\ref{app:proj} that any operator that commutes with all unitaries $U(g)$ has the form of Eq.~\ref{G-inv operators}, which implies that any
G-invariant unitary is of the form \cite{BRS07},
\begin{equation} \label{G-inv unitaries}
V_{\textnormal{G-inv}}=\bigoplus_\mu \mathbb{I}_{\mathcal{M}_\mu} \otimes V_{\mathcal{N}_\mu},
\end{equation}
where  $V_{\mathcal{N}_\mu} $ acts unitarily on ${\mathcal{N}_\mu} $.

Now suppose state $|\psi\rangle$ can be transformed to another state $|\phi\rangle$ by a G-invariant unitary $V_{\textnormal{G-inv}}$. Then given Eq.~\ref{G-inv unitaries}, it follows that for all $\mu$,
\begin{equation}
\Pi_\mu |{\phi}\rangle= \Pi_\mu V_{\textnormal{G-inv}} |{\psi}\rangle= \mathbb{I}_{\mathcal{M}_\mu} \otimes V_{\mathcal{N}_\mu} \Pi_\mu |{\psi}\rangle
\end{equation}
Eq.~(\ref{reduction}) then follows from the cyclic property of the trace and the unitarity of $V_{\mathcal{N}_\mu}$.

Now we prove the converse. If Eq.~(\ref{reduction}) holds,
then there exists a G-invariant unitary which transforms  $|{\psi}\rangle$ to $|{\phi}\rangle$.
First note that we can think of the two vectors  $\Pi_\mu  |{\psi}\rangle$ and $\Pi_\mu  |{\phi}\rangle$ as two different purifications of $\textnormal{tr}_{\mathcal{N}_\mu}(\Pi_\mu  |{\psi}\rangle\langle \psi|  \Pi_\mu)=\textnormal{tr}_{\mathcal{N}_\mu}( \Pi_\mu |{\phi}\rangle\langle \phi| \Pi_\mu )$. So   $\Pi_\mu   |{\psi}\rangle$ can be transformed to $\Pi_\mu |{\phi}\rangle$  by a unitary acting on $\mathcal{N}_\mu$, denoted by $V_{\mathcal{N}_\mu}$, such that
\begin{equation}
\mathbb{I}_{\mathcal{M}_\mu} \otimes V_{\mathcal{N}_\mu}   \Pi_\mu  |{\psi}\rangle=\Pi_\mu |{\phi}\rangle
\end{equation}
 (See e.g.\cite{Nie00}). By defining
\begin{equation}
V\equiv \bigoplus_\mu \mathbb{I}_{\mathcal{M}_\mu} \otimes V_{\mathcal{N}_\mu}
\end{equation}
we can easily see that $V$ is a G-invariant unitary and moreover $V|{\psi}\rangle=|{\phi}\rangle$. This completes the proof. 
\end{proof}


For an arbitrary state $\rho$ we call the set of operators $\{ \textnormal{tr}_{ \mathcal{N}_{\mu}} ({\Pi}_{\mu}\  \rho \ {\Pi}_{\mu} )\}$, the \emph{reduction onto irreps} of $\rho$. So in the above theorem we have proven that the unitary G-equivalence class of a pure state is totally specified by its reduction onto irreps.  Note, however, that as we will see in Sec.~\ref{sec:tworepns}, this is not true for general mixed states.

\begin{example} \label{Ex-U(1)-reduc}
Recall the quantum optics example studied in section \ref{q-optics-example} where the set of all phase shifts forms a representation of group U(1). There the representation of group U(1) is
$e^{i\theta}\rightarrow U(\theta)$ where  the phase  shift operator $U(\theta)$ is
\begin{equation}
U(\theta)\equiv e^{i N\theta}= \sum_{n} e^{i n \theta } \sum_{\alpha} |n,\alpha\rangle\langle n,\alpha|,
\end{equation}
where $N$ is the number operator with integer eigenvalues such that $N|n,\alpha\rangle=n|n,\alpha\rangle$ and where $\alpha$ is a multiplicity index. In this case all irreps are one-dimensional.   It follows that the reduction onto irreps of a pure state $|\psi\rangle=\sum_{n,\alpha} \psi_{n,\alpha} |n,\alpha\rangle $ is simply given by
\begin{equation}
p_{\psi}(n) \equiv \langle \psi| \Pi_n |\psi \rangle=\sum_{\alpha} |\psi_{n,\alpha}|^{2},
\end{equation}
where $\Pi_{n}$ is the projector to the eigen-subspace corresponding to the eigenvalue $n$ of $N$.  That is, the reduction onto irreps is the probability distribution over the spectrum of the number operator induced by $|\psi\rangle$.  Consequently, two pure states are unitarily U(1)-equivalent if and only if they define the same probability distribution over number. 
\end{example}

\subsection{The information-theoretic characterization: equality of characteristic functions} \label{section of char}

We will show that by taking the information-theoretic point of view, one finds that the unitary G-equivalence class of a pure state is specified entirely by its characteristic function, which is defined as follows.
\begin{definition} [Characteristic function]
The \emph{characteristic function} of a state $\rho$ relative to a projective unitary representation $\{U(g): g\in G\}$ of a group G is a function $\chi_\rho : G \to \mathbb{C}$ of the form
\begin{equation}
\chi_\rho(g)\equiv\textnormal{tr}\left(\rho U(g)\right)
\end{equation}
\end{definition}
Specifically, we have
\begin{theorem} \label{theorem of characteristic}
Two pure states $|\psi\rangle$ and $|\phi\rangle$ are unitarily
G-equivalent if and only if their characteristic functions are equal,
\begin{equation} \label{equality_characteristics}
\forall g\in \textnormal{G}:\ \ \langle\psi|U(g)|\psi\rangle=\langle\phi|U(g)|\phi\rangle .
\end{equation}
\end{theorem}

The benefit of trying to characterize the G-equivalence classes using the information-theoretic perspective is that we can make use of known results concerning the unitary interconvertibility of sets of pure states.  We express the condition for such interconvertibility as a lemma, after recalling the definition of the Gram matrix of a set of states.

\begin{definition}[Gram matrix] Consider the set of states $\{|{\psi_\theta}\rangle\}$.  If $\theta$ is a discrete parameter, then we define the \emph{Gram matrix} of the set $\{|{\psi_\theta}\rangle\}$ by $X_{\theta,\theta'}  \equiv\langle{\psi_\theta}|{\psi_{\theta'}}\rangle$. If $\theta$ is a continuous parameter, then we can define the function $X({\theta,\theta'})\equiv\langle{\psi_\theta}|{\psi_{\theta'}}\rangle$, which, with a slight abuse of terminology, we will also call the \emph{Gram matrix} of the set $\{|{\psi_\theta}\rangle\}$.
\end{definition}

\begin{lemma}\label{Gram2}
There exists a unitary operator $V$ which transforms  each member of $\{|{\psi_\theta}\rangle\}$ to its corresponding member in $\{|{\phi_\theta}\rangle\}$, that is, $\forall \theta:\  V|{\psi_\theta}\rangle=|{\phi_\theta}\rangle$, if and only if the Gram matrices of the two sets of states are equal, i.e. $$\forall \theta,\theta':\ \  \langle{\psi_\theta}|{\psi_{\theta'}}\rangle=\langle{\phi_\theta}|{\phi_{\theta'}}\rangle$$
\end{lemma}
A simple proof of this lemma is provided in the footnote.\footnote{
The necessity of the equality of the Gram matrices is trivial. Sufficiency is proven as follows. Suppose we use a subset $\{|\psi_{\theta_{1}}\rangle, |\psi_{\theta_{2}}\rangle,\cdots   \}$ of $\{|{\psi_\theta}\rangle\}$ to  build an orthonormal basis for the subspace spanned by $\{|{\psi_\theta}\rangle\}$ via  the Gram-Schmidt process and call this basis I. Similarly,  use  the subset  $\{|\phi_{\theta_{1}}\rangle, |\phi_{\theta_{2}}\rangle,\cdots   \}$ of $\{|{\phi_\theta}\rangle\}$ to  build an orthonormal basis for the subspace spanned by $\{|{\phi_\theta}\rangle\}$  via  the Gram-Schmidt process and call this basis II.   Recall that the Gram-Schmidt orthogonalization process depends only on the Gram matrix of the set of states.  Since,  by assumption, the Gram matrix of the two sets of states are equal then for any state $|\psi_{\theta}\rangle\in \{|{\psi_\theta}\rangle\}$ its description in basis I is the same as the description of the corresponding $|\phi_{\theta}\rangle\in \{|{\phi_\theta}\rangle\}$ in basis II.  It follows that if $V$ is the unitary which transforms basis I to basis II, then by linearity for all $|\psi_{\theta}\rangle\in \{|\psi_{\theta}\rangle\}$, $V$ maps $|\psi_{\theta}\rangle$ to the state $|\phi_{\theta}\rangle$. This proves the lemma.
}

It is now straightforward to prove theorem \ref{theorem of characteristic}.

\begin{proof} [Proof of theorem \ref{theorem of characteristic}]
 By definition \ref{unitary G-equivalence I}, $|{\psi}\rangle$ and $|{\phi}\rangle$ are unitarily G-equivalent if there exists a unitary transformation $V_{\textnormal{G-inv}}$ which takes $|{\psi}\rangle$ to $|{\phi}\rangle$. By lemma \ref{unitary} there exists such a unitary if and only if there exists a unitary $V$ such that $\forall g\in G:\ VU(g)|\psi\rangle=U(g)|\phi\rangle$. By lemma (\ref{Gram2}), the necessary and sufficient condition for the existence of such a unitary is the equality of the Gram matrices of the set  $\{U(g)|{\psi}\rangle:g\in G\}$  and the set $\{U(g)|{\phi}\rangle :g\in G\}$.  Given that  the elements of these matrices are, respectively,
$$[X_\psi]_{g_1,g_2}=\langle\psi|U^\dag(g_1)U(g_2)|{\psi}\rangle=\omega(g_{1}^{-1},g_{2}) \langle\psi| U(g^{-1}_1g_2)|{\psi}\rangle,$$
and
$$[X_\phi]_{g_1,g_2}=\langle\phi|U^\dag(g_1)U(g_2)|{\phi}\rangle=\omega(g_{1}^{-1},g_{2}) \langle\phi|U(g^{-1}_1g_2)|{\phi}\rangle,$$
where we have used the fact $g\rightarrow U(g)$ is a projective unitary representation and so  $$U^\dag(g_1)U(g_2)=U(g^{-1}_{1})U(g_2)=\omega(g^{-1}_{1},g_{2})$$
for the cocycle $\omega$.  Equality of the Gram matrices is equivalent to
\begin{equation}\label{proof-lem-char}
\forall g\in \textnormal{G}:\ \ \langle\psi|U(g)|\psi\rangle=\langle\phi|U(g)|\phi\rangle,
\end{equation}
and this is simply the statement that the characteristic functions of $\psi$ and $\phi$ are equal.
\end{proof}



\begin{example}
In example \ref{Ex-U(1)-reduc} we found the  characterization of unitary equivalence classes based on the reduction of states to irreps  in the case of group U(1) with representation $e^{i\theta}\rightarrow e^{i\theta N}$ where $N$ is the number operator with nonnegative integer eigenvalues.
Here, we use the result of lemma \ref{Gram2} to find another characterization  of these unitary equivalence classes in terms of characteristic functions of states. In this case, for arbitrary state $|\psi\rangle=\sum_{n,\alpha} \psi_{n,\alpha} |n,\alpha\rangle $, the characteristic function is given by the expectation value of the phase shift operator, i.e.,
\begin{equation}
\chi_{\psi}(\phi) \equiv \langle \psi | \exp(i\phi {N}) | \psi \rangle= \sum_n p_{\psi}(n) e^{i n\phi},
\end{equation}
 where $p_{\psi}(n) = \sum_{\alpha} |\psi_{n,\alpha}|^2$ is the reduction onto irreps.
\end{example}
It follows that in the U(1) case, the reduction onto irreps and the characteristic function are related by a Fourier transform.
The Fourier transform can also be defined for arbitrary compact Lie groups or for finite groups (which might be non-Abelian) and in these cases as well, it describes  the relation between the reduction onto irreps and the characteristic function, as will be shown in section \ref{Sec-Reduction}.


\subsection{Approximate notion of unitary G-equivalence} \label{sec-approximate}

We have found the necessary and sufficient condition for the existence of a G-invariant  unitary which transforms a pure state $\psi$ to another pure state $\phi$. This is the condition for exact transformation. But there might be situations in which we cannot transform $\psi$ to $\phi$ but we can transform it  to some state close to $\phi$.

In the following we demonstrate that if the reductions onto irreps of two pure states $\psi$ and $\phi$ are close in some sense (or equivalently their characteristic functions are close) then there exists a G-invariant unitary which transforms $\psi$ to a state close to $\phi$ (See appendix \ref{app:proofofapproxGequivalence} for a discussion about the relevant notion of distance in this context).

Recall that the fidelity of two positive operators $A_{1}$ and $A_{2}$ is defined as
\begin{equation}
\textnormal{Fid}(A_{1},A_{2})\equiv  \|\sqrt{A_1}\sqrt{A_2}\|=\textnormal{tr}(\sqrt{  \sqrt{A_{1}}  A_{2} \sqrt{A_{1} }  })
\end{equation}
 where $\|\cdot\|$ denotes the trace norm. 

\begin{theorem} \label{approximate_theorem}
Suppose  $\{F_1^{(\mu)}\}$ and $\{F_2^{(\mu)}\}$ are respectively the reductions onto irreps of  $\psi_1$ and $\psi_2$,  two arbitrary pure states in the same Hilbert space. Then for any G-invariant unitary $V$
 acting on this space
\begin{equation}\label{approximate_trans}
|\langle\psi_2|V|\psi_1\rangle|\leq \sum_\mu \textnormal{Fid}(F_{1}^{(\mu)},F_{2}^{(\mu)})
\end{equation}
Furthermore there exists a G-invariant unitary $V$ for which the equality holds.
\end{theorem}
According to this theorem if the fidelities of the reductions onto irreps is high then there exists a G-invariant unitary which transforms one of the states to a state very close to the other. On the other hand, if these fidelities are low we can never transform one of the states to a state close to the other via G-invariant unitaries.
\begin{remark}\label{remark-unit}
For $\{F_1^{(\mu)}\}$ and $\{F_2^{(\mu)}\}$  the reductions of an arbitrary pair of states   it holds that
 $\sum_\mu \textnormal{Fid}(F_{1}^{(\mu)},F_{2}^{(\mu)})\le 1$ and the equality holds  iff  $\forall \mu: F_{1}^{(\mu)}=F_{2}^{(\mu)}$.
So theorem \ref{Thm-reductions} is a special case of theorem \ref{approximate_theorem}.
\end{remark}
We present the proof of theorem \ref{approximate_theorem} as well as some  other versions of it and the proof of remark \ref{remark-unit} in appendix \ref{app:proofofapproxGequivalence}.

\begin{example}\label{example-charac}
Recall our quantum optics example where the set of all phase shifts forms a representation of the group U(1) (see example \ref{Ex-U(1)-reduc}).  Let  $p_{\psi}$ and $p_{\phi}$ be the probability distributions over integers  which describe the reductions onto irreps of the states $\psi$ and $\phi$ respectively.  Then theorem~\ref{approximate_theorem} implies that for any U(1)-invariant unitary $V$,
\begin{equation}
|\langle\psi|V|\phi\rangle|\le \sum_{n} \sqrt{p_{\psi}(n)p_{\phi}(n)}
\end{equation}
and furthermore there exists a U(1)-invariant unitary for which the equality holds.
\end{example}

\section{What are the reduction onto irreps and the characteristic function?} \label{Sec-Reduction}

We have found two different characterizations of the unitary  G-equivalence class of pure states, namely the characteristic function of states and the reduction onto irreps of states.  In this section, we will show that the reduction onto irreps and the characteristic function are simply two particular representations of the reduction of the state to the associative algebra (for the degree of freedom associated to the symmetry transformation)  and that these representations are related to one another by a generalized Fourier transform. Furthermore, we provide a list of properties of characteristic functions which will be useful in the rest of this section.

In appendices \ref{app:distinguish} and \ref{app:charfuncs} we present more discussions about the meaning of characteristic functions of states. In appendix \ref{app:distinguish} we discuss about the interpretation of the absolute value of the characteristic function of state $\psi$,
$$\left|\chi_{\psi}(g)\right|=\left\langle \psi| U(g)|\psi\rangle\right|,$$
 in terms of the pairwise distinguishability of states in the set $\{U(g)|\psi\rangle: g\in G\}$. In particular, we argue that though the function $\left|\chi_{\psi}(g)\right|$ uniquely specifies all the pairwise distinguishabilities in this set, nevertheless it cannot specify the information that can be transferred using   the encoding $g\rightarrow U(g)|\psi\rangle$ and so it can not specify the asymmetry of state $\psi$. Also, in appendix \ref{app:charfuncs} we show that the characteristic function of a quantum state can be thought as a natural generalization of the notion of the characteristic function of a probability distribution.

\subsection{Two representations of the reduction to the associative algebra}\label{sec:tworepns}

If we are interested in only some particular degree of freedom of a quantum system then we do not need the full description of the state in order to infer the statistical features (expectation values, variances, correlations between two different observables, etcetera) of that degree of freedom.  In particular suppose we are interested in the statistical properties of the set of operators $\{O_i\in\mathcal{B}(\mathcal{H})\}$.  Closing this set under the operator product and sum yields the associative algebra generated by $\{O_i\}$, which is the set of all polynomials in $\{O_i\}$.  We denote this associative algebra by $\textnormal{Alg}\{O_i\}$.  To specify all the statistical properties of the state $\rho\in\mathcal{B}(\mathcal{H})$ for the set of observables $\{O_i\}$ it is necessary and sufficient to specify the expectation values of all the operators in $\textnormal{Alg}\{O_i\}$ under the state $\rho$. The object that contains all and only this information is called the reduction of the state to the associative Algebra, denoted $\rho |_{\textnormal{Alg}\{O_i\}}$.

$\textnormal{Alg}\{O_i\}$, considered as a finite-dimensional $C^*$-algebra,
 has a unique decomposition (up to unitary equivalence) of the form
\begin{equation}\label{Algebra_decomp}
\bigoplus_J \mathfrak{M}_{m_J} \otimes \mathbb{I}_{n_J}
\end{equation}
where $\mathfrak{M}_{m_J}$ is the full matrix algebra $\mathcal{B}(\mathbb{C}^{m_J})$ and $\mathbb{I}_{n_J}$ is the identity on $\mathbb{C}^{n_J}$ \cite{Davidson}.
This means that any element $A$ of the algebra can be written as
\begin{equation}
A= \bigoplus_J  A^{(J)}\otimes \mathbb{I}_{n_J}
\end{equation}
where $A^{(J)}\in \mathcal{B}(\mathbb{C}^{m_J})$.  Furthermore, if we consider the set of all elements of the algebra, that is, all $A\in \textnormal{Alg}\{O_i\}$,
and look at the set of corresponding $A^{(J)}$ for fixed $J$,  this set of operators acts irreducibly on $\mathbb{C}^{m_J}$ and spans $\mathcal{B}(\mathbb{C}^{m_J})$. Clearly this decomposition induces the following structure on the Hilbert space
\begin{equation}
\mathcal{H}=\bigoplus_J \mathcal{M}_{J} \otimes \mathcal{N}_{J}.
\end{equation}
where $\mathcal{M}_{J}$ is isomorphic to $\mathbb{C}^{m_J}$ and  $\mathcal{N}_{J}$ is isomorphic to  $\mathbb{C}^{n_J}$.

 Suppose $\Pi_J$ is the projective operator to the subspace $\mathcal{M}_{J} \otimes \mathcal{N}_{J}$. Then to specify all the relevant information about the observables in the Algebra  for the given state $\rho$ it is necessary and sufficient to know all of the operators
\begin{equation} \label{reduc-sectors}
\rho^{(J)}\equiv tr_{\mathcal{N}_J}(\Pi_J \rho \Pi_J).
\end{equation}
Then for any arbitrary observable $A$ in the Algebra we have
\begin{equation}
\textnormal{tr}(A\rho)= \sum_{J} \textnormal{tr}(A^{(J)}\rho^{(J)})
\end{equation}
and so specifying the set $\{\rho^{(J)}\}$ we know all the relevant information about the state. In other words,  $\{\rho^{(J)}\}$ uniquely specifies the reduction to the Algebra $\rho |_{\textnormal{Alg}\{O_i\}}$.

The above discussion applies to any arbitrary set of observables.  Here, we will be interested in the case where this set describes the degree of freedom associated to some symmetry transformation.  If the symmetry transformation is associated with the symmetry group $G$ and projective unitary representation $\{ U(g):g\in G\}$ on the Hilbert space of the system, then the set of observables to consider are all those in the linear span of $\{ U(g):g\in G\}$. In particular, in the case of Lie groups  this set contains the representation of all generators of  the Lie Algebra (associated to the group) and all the polynomials formed by these generators. For example, in the case of  \textnormal{SO(3)} the set includes  all the observables in the linear span of $\{ U(\Omega): \Omega\in SO(3)\}$ and so it clearly contains all the generators, which in this case are angular momentum operators, as well as all polynomials of these.

Decomposition  of this algebra in the form of Eq.~(\ref{Algebra_decomp}) in fact coincides with  the decomposition of the unitary projective representation $\{ U(g):g\in G\}$ to irreps
\begin{equation}\label{irrepdecompofU}
U(g)\cong \bigoplus_\mu U^{(\mu)}(g)\otimes \mathbb{I}_{N_\mu}
\end{equation}
where $\mu$ labels the irreps and $\mathbb{I}_{N_\mu}$ is the identity acting on the multiplicity subsystem associated to irrep $\mu$ (Remember that $G$ is by assumption a finite or compact Lie group and so it is completely reducible.). Here we can think of $\mu$ playing the same role as $J$ in the decomposition of the arbitrary Algebra in Eq.~(\ref{Algebra_decomp}). Each irrep index $\mu$ appearing in the decomposition of $\{ U(g):g\in G\}$ corresponds to one J in Eq.~(\ref{Algebra_decomp}) and the set $\{U^\mu(g):g\in G\}$ for a fixed $\mu$ spans the full matrix algebra $\mathfrak{M}_{m_J}$ of the corresponding $J$.  Consequently, the spaces on which the projective unitary representation of G acts irreducibly are simply the $\mathcal{M}_{J}$.  So it follows that in this case, where the associative Algebra coincides with the span of the elements of the projective unitary representation of the group, $\{ U(g):g\in G\}$, the set of operators $\{\rho^{(J)}\}$ (defined by Eq.(\ref{reduc-sectors})) is simply the reduction onto the irreps of the state $\rho$, the generalization to mixed states of the notion defined in the section \ref{reductions onto irreps}, and therefore we can conclude that the reduction onto the irreps is a representation of the reduction onto the associative algebra.



Another way to specify the reduction of the state onto the associative algebra is to specify the Hilbert-Schmidt inner product of $\rho$ with each of the $U(g)$, namely, $\textnormal{tr}(\rho U(g))$ for all $g\in$ G. So if we define the characteristic function associated to the state $\rho$ as the function $\chi_\rho: G \to \mathbb{C}$ defined by
$\chi_\rho(g) \equiv \textnormal{tr}(\rho U(g)),$
then the characteristic function is a particular representation of the reduction to the associative algebra.  It is clear that this definition constitutes a generalization to mixed states of the notion of characteristic functions introduced in the section \ref{section of char}.


To summarize, we have
\begin{remark} For a state $\rho \in \mathcal{B}(\mathcal{H})$ and a projective unitary representation $U$ of a group $G$, the reduction of $\rho$ to the associative algebra $\textnormal{Alg}\{ U(g): g \in G \}$ can be represented either in terms of the \emph{reduction onto irreps} of $\rho$, defined as
\begin{equation}
\{ \rho^{(\mu)} \equiv \textnormal{tr}_{\mathcal{N}_{\mu}}(\Pi_{\mu} \rho \Pi_{\mu})\} ,
\end{equation}
(where the Hilbert space decomposition induced by $U$ is $\mathcal{H}= \bigoplus_{\mu} \mathcal{M}_{\mu} \otimes \mathcal{N}_{\mu}$ and $\Pi_{\mu}$ projects onto $\mathcal{M}_{\mu} \otimes \mathcal{N}_{\mu}$),
or in terms of the \emph{characteristic function} of $\rho$, defined as
\begin{equation}\label{charfunc_general}
\chi_\rho(g) \equiv \textnormal{tr}(\rho U(g)).
\end{equation}
\end{remark}

Finally, we note that the relationship between these two representations is the Fourier transform over the group.
\begin{proposition}
The characteristic function and reduction onto irreps can be computed one from the other via
\begin{equation} \label{reduction_char}
\chi_\rho(g)=\sum_\mu \textnormal{tr}(\rho^{(\mu)}U^{(\mu)}(g))
\end{equation}
and
\begin{equation} \label{Fourier}
\rho^{(\mu)}=d_\mu \int dg \chi_\rho(g^{-1}) U^{(\mu)}(g).
\end{equation}
\end{proposition}

\begin{proof}
The expression for $\chi_\rho(g)$ in terms of $\{\rho^{(\mu)}\}$, Eq. (\ref{reduction_char}), follows directly from Eqs. (\ref{irrepdecompofU}) and (\ref{charfunc_general}).
Conversely, to find the $\{\rho^{(\mu)}\}$ in terms of $\chi_\rho(g)$ we use the Fourier transform over the group. The idea is based on the following orthogonality relations which are part of the Peter-Weyl theorem (See e.g. \cite{Barut-1986}):
\begin{equation} \label{orthogonalityintext}
\int_G dg U_{i,j}^{(\mu)}(g) {{\overline{U}^{(\nu)}_{k,l}}}(g)=\frac{\delta_{\mu,\nu}\delta_{i,k}\delta_{j,l}}{ d_\mu}
\end{equation}
where $\{U_{i,j}^{\mu}\}$ are the matrix elements of $U^{\mu}(g)$, $dg$ is the unique Haar measure on the group, bar denotes the complex conjugate and $d_\mu$ is the dimension of irrep $\mu$. According to this theorem any continuous function on a compact Lie group can be uniformly approximated by linear combinations of matrix elements $U_{i,j}^{(\mu)}(g) $.   Note that for the finite groups, we can get the same orthogonality relations by replacing the integral with a summation. Furthermore any function over a finite group can be expressed as a linear combination of the matrix elements of irreps. So basically all the properties we use hold for finite groups as well as compact Lie groups.

An arbitrary operator $A^{(\mu)}$ in $\mathcal{B}(\mathcal{M}_{\mu})$  can be written as a linear combination of elements of $\{U^{(\mu)}(g):\ g\in G\}$. The above orthogonality relations imply that this expansion  has a simple form as
\begin{equation}
A^{(\mu)}=d_{\mu}\int dg\ U^{\mu}(g)\ \ \textnormal{tr}(A^{(\mu)} U^{\mu}(g^{-1}))
\end{equation}
Clearly this can be considered as a completeness relation where we have decomposed the identity map on $\mathcal{B}(\mathcal{M}_{\mu})$ as the sum of projections to a basis (which is generally overcomplete).    Also note that the orthogonality relations imply that for $\nu \neq \mu$
\begin{equation}
\int dg\ U^{\nu}(g) \textnormal{tr}(A^{(\mu)} U^{\mu}(g^{-1}))=0\ \ \ \ \ (\nu \neq \mu)
\end{equation}
Using these orthogonality relations, we obtain Eq.~(\ref{Fourier}).
\end{proof}

We should emphasize that the reduction onto the associative algebra, though sufficient for deciding G-equivalence of pure states, is not in general sufficient for deciding G-equivalence of arbitrary states, i.e., mixed and pure.  Its sufficiency in the case of pure states follows from its sufficiency for deciding unitary G-equivalence (proven in Sec. \ref{section of char}) and the fact that the unitary G-equivalence classes are a fine-graining of the G-equivalence classes.  Its insufficiency in the case of mixed states can be established by the following simple example of two states (one pure and one mixed) that have the same characteristic function but fall in different G-equivalence classes.  The example is for the case of U(1)-covariant operations, and the two states are $\frac{1}{2}(|0\rangle +|1\rangle)(\langle0|+\langle1|)$ and $\frac{1}{2}(|0\rangle\langle 0| +|1\rangle \langle 1|)$.  The second is clearly U(1)-invariant while the first is not and so they must lie in different U(1)-equivalence classes.  Nonetheless, the characteristic function for both equals $\chi(\theta)=1/2(1+\exp(i\theta))$.

We close this section by  mentioning another consequence of the orthogonality relations Eq.~(\ref{orthogonalityintext}) which is useful later. Suppose $A,B$ are arbitrary operators in $\mathcal{B}(\mathcal{M}_{\mu})$ and
$$\chi_A(g)\equiv \textnormal{tr}(AU^{(\mu)}(g)),\ \  \chi_B(g)\equiv \textnormal{tr}(BU^{(\mu)}(g)),\ \ \textnormal{and}$$ 
$$ \chi_{AB}(g)\equiv \textnormal{tr}(ABU^{(\mu)}(g))$$ are respectively the characteristic functions of $A,B$  and $AB$. Then
\begin{equation} \label{convolution}
\chi_{AB}=d_\mu\ \chi_{A} \ast \chi_{B}
\end{equation}
where $\ast$ is the convolution of two functions \footnote{Note that for non-Abelian groups $f_1\ast f_2$ is not necessarily equal to $f_2\ast f_1$.}
\begin{equation} \label{def-conv}
f_1\ast f_2(g)\equiv\ \int dh f_1(gh^{-1})f_2(h)
\end{equation}
 In particular,   since $\textnormal{tr}(AB)=\chi_{AB}(e)$ (where $e$ is the identity of the group) the above formula can be used to find $\textnormal{tr}(AB)$ in terms of the characteristic functions of $A$ and $B$. Using Eq.(\ref{convolution}) we get
\begin{align*}
\textnormal{tr}(AB)=\chi_{AB}(e)&=d_\mu\ [\chi_{A} \ast \chi_{B}](e)\\&=d_\mu\ \int dh\  \chi_A(h)  \chi_B(h^{-1})
\end{align*}


\subsection{Properties of characteristic functions} \label{properties}

The characteristic functions introduced here are quantum analogues of those used in classical probability theory. The connection is discussed in detail in Appendix \ref{app:charfuncs}. Here we simply summarize some useful properties of characteristic functions.

\begin{enumerate}
  \item A function $\phi(g)$ from the finite or compact Lie group  $G$ to complex numbers  is the characteristic function of a physical state iff it is continuous (in the case of Lie groups)  positive definite (as defined in  appendix \ref{app:charfuncs}) and normalized (i.e. $\phi(e)=1$ where $e$ is the identity of the group.). (This property assumes that all irreps are physically accessible. ) \label{Bochner}

  \item  The characteristic function of a state is invariant under G-invariant unitaries acting on that state, 
      $$\chi_{\mathcal{V}_{\textnormal{G-inv}}[\rho]}(g)=\chi_{\rho}(g),$$
      where $\mathcal{V}_{\textnormal{G-inv}}[\cdot]=V_{\textnormal{G-inv}}(\cdot)V^{\dag}_{\textnormal{G-inv}}$ and $[V_{\textnormal{G-inv}},U(g)]=0$ for all $g\in G$.
 \item Characteristic functions multiply under tensor product, \label{multi}
      \begin{equation} \chi_{\rho\otimes\sigma}(g)= \chi_{\rho}(g)\chi_{\sigma}(g). \end{equation}
  \item $|\chi_\rho(g)|\leq 1$ for all $g\in G$ and $\chi_\rho(e)=1$ where $e$ is the identity of the group.
  \item If $|\chi_\rho(g_s)|=1$ for $g_s\in G$ then $g_s$ is a symmetry of $\rho$.  If $\rho$ is a pure state, then $g_s$ is a symmetry of $\rho$  if and \emph{only if}  $|\chi_\rho(g_s)|=1$. \label{symmetric states}
  \item So $|\chi_\rho(g)|=1$ for all $g\in G$ implies that the state is invariant; in this case  $\chi_{\rho}(g)$ is a 1-d representation of the group. \label{tot-Symmetric}

\item Suppose $L$ is the representation of a generator of a Lie group on the Hilbert space of a system such that $\{e^{i\theta L}:\theta\in (0,2\pi]\}$   is the representation of a U(1)-subgroup of the group. Then we can find all moments of $L$ using the characteristic function.
\begin{equation}
\textnormal{tr}(\rho L^{k})=i^{-k}\frac{\partial^k}{\partial\theta^k} \chi_{\rho}(e^{i\theta L} ) \mid_{\theta=0}
\end{equation} \label{derivatives}
(Note that by $\chi_{\rho}(e^{i\theta L})$ we really mean  $\chi_{\rho}(g)$  for the group element $g\in G$ which is represented by $e^{i\theta L}$.)


  \end{enumerate}
\begin{proof}
Item \ref{Bochner} 
is proven in Appendix \ref{app:qcharfuncs}. 
All the rest of these properties can simply be proved by using the definition of the characteristic function, $\chi_{\rho}(g)=\textnormal{tr}(\rho U(g))$, and group representation properties. For example to  prove item \ref{multi} we use the fact that if the representation of the symmetry G on the systems $A$ (with state $\rho$) and $B$ (with state $\sigma$) are $g\rightarrow U_{A}(g)$ and $g\rightarrow U_{B}(g)$ then the representation of the symmetry on the joint system $AB$ is  $g\rightarrow U_{A}(g)\otimes U_{B}(g)$. Then
\begin{align*}
\chi_{\rho\otimes \sigma}(g)&=\textnormal{tr}\left(\rho\otimes\sigma U_{A}(g)\otimes U_{B}(g)\right)\\ &=\textnormal{tr}\left(\rho U_{A}(g)\right)\textnormal{tr}\left(\sigma U_{B}(g)\right)\\ &= \chi_{\rho}(g) \chi_{\sigma}(g)
\end{align*}

To prove \ref{symmetric states}  we note that if $|\chi_\rho(g_s)|=1$ for $g_s\in G$ then all eigenvectors of $\rho$ are eigenvectors of $U(g_{s})$ with the same eigenvalue. As a result we get $[\rho,U(g_s)]=0$ and so the state has the symmetry $g_s$.  On the other hand, $[\rho,U(g_s)]=0$ does not imply that $|\chi_\rho(g_s)|=1$.  For instance, the state $\frac{1}{2}|0\rangle\langle 0|+\frac{1}{2}|1\rangle\langle 1|$ where $|n\rangle$ is a number eigenstate is U(1)-invariant, but nonetheless, for $\phi \ne 0$, $|\chi_\rho(\phi)|\ne 1$.  Therefore the points for which the amplitude of the characteristic function is one are a subset of the symmetries of the state.  Meanwhile, if a \emph{pure} state $|\psi\rangle$ has symmetry $g_s$, such that $U(g_{s})|\psi\rangle=e^{i\theta}|\psi\rangle$ for some $\theta$, then obviously $|\chi_\psi(g_s)|=1$. So for pure states the points for which the amplitude of the characteristic function is one are exactly the state's symmetries.

To prove \ref{tot-Symmetric}, we first note that if $|\chi_\rho(g)|=1$ for all $g\in G$, then the symmetry subgroup of $\rho$ is the entire group $G$, which is the definition of $\rho$ being G-invariant. Furthermore, for each $g$, the eigenvectors of $\rho$ all live in the same eigenspace of $U(g)$. Since the eigenvalue of a unitary is a phase factor, each such eigenvector $|\nu\rangle$ must satisfy $U(g)|\nu\rangle=e^{i{\theta(g)}}|\nu\rangle$ for some phase $e^{i{\theta(g)}}$.  It is then clear that $\chi_{\rho}(g)=e^{i{\theta(g)}}$ and is a 1-dimensional representations of the group. 
\end{proof}

Among the above properties, the fact that the tensor product of states is represented by the product of their characteristic functions (property \ref{multi}) turns out to be particularly useful.  This is because the alternative representation, in terms of reductions onto irreps, does not provide a simple expression for the reduction of $\rho\otimes \sigma$ in terms of the reduction of $\rho$ and the reduction of $\sigma$.  It involves Clebsch-Gordon coefficients and is generally quite complicated for non-Abelian groups.

For this and other reasons, the characteristic function is generally our preferred way of representing the reduction of the state onto the algebra, and consequently we will make heavy use of it  to answer various questions about the manipulation of asymmetry of pure states. 

\section{G-equivalence classes} \label{G-equiv}

We have seen that the characteristic function of a pure state uniquely specifies its unitary G-equivalence class. However, it is $G$-equivalence
rather than unitary $G$-equivalence that implies that two states have the same asymmetry properties, so we must ultimately characterize the former.
Fortunately, for compact connected Lie groups, the conditions under which two pure states are G-equivalent can also be stated simply in terms of their characteristic
functions, as is shown presently. \strut


\begin{theorem}
\label{maintheorem}
For $G$ a compact connected Lie group, two pure states $|\psi\rangle$ and $|\phi\rangle$
are $G$-equivalent (i.e. they can be reversibly interconverted one to the
other by $G$-covariant operations) iff there exists a 1-dimensional
representation of $G$, $e^{i\Theta(g)}$, such that
\begin{equation} \label{ConditionForGequivalence}
\forall g\in G:\ \ \langle\psi|U(g)|\psi\rangle=e^{i\Theta(g)}\langle
\phi|U(g)|\phi\rangle.
\end{equation}
\end{theorem}
Since  semi-simple compact Lie groups do not have any non-trivial 1-dimensional representation, the above theorem implies
\begin{corollary}
For $G$ a semi-simple compact connected Lie group, two pure states $|\psi\rangle$ and $|\phi\rangle$
are $G$-equivalent iff their characteristic functions are equal, i.e.
\begin{equation}
\forall g\in G:\ \ \langle\psi|U(g)|\psi\rangle=\langle
\phi|U(g)|\phi\rangle.
\end{equation}
\end{corollary}

The above theorem applies only to compact connected Lie groups. By putting a restriction on the states we can prove a similar theorem  which applies to both compact Lie groups and finite groups

\begin{theorem} \label{general-deterministic}
Two pure states $|\psi\rangle$ and $|\phi\rangle$ for which $\langle\psi|U(g)|\psi\rangle$ and $\langle
\phi|U(g)|\phi\rangle$ are nonzero for all $g\in G$
are $G$-equivalent (i.e. they can be reversibly interconverted one to the
other by G-covariant operations) iff there exists a 1-dimensional
representation of $G$, $e^{i\Theta(g)}$, such that
\begin{equation}
\forall g\in G:\ \ \langle\psi|U(g)|\psi\rangle=e^{i\Theta(g)}\langle
\phi|U(g)|\phi\rangle.
\end{equation}
\end{theorem}

\begin{proof}(Theorems \ref{maintheorem} and \ref{general-deterministic})
The main tool we use in this proof is the Stinespring dilation theorem for G-covariant channels discussed in the preliminaries (See \cite{thesis:Marvian} and \cite{Werner}). According to this result any G-covariant channel can be implemented by preparing an environment in a G-invariant state and coupling it to the system with a G-invariant unitary.

First we prove that Eq.~(\ref{ConditionForGequivalence}) implies that $|\psi\rangle$ and $|\phi \rangle$ are G-equivalent. Suppose $|\nu_{0}\rangle$ is a G-invariant state of the environment whose characteristic function is constant and equal to 1 for all group elements and   $|\nu\rangle$  is a state of the environment with characteristic function $e^{i\Theta(g)}$ where by assumption $e^{i\Theta(g)}$ is a 1-dimensional representation of the group (such states always exist by virtue of property \ref{Bochner} of characteristic functions listed in section \ref{properties}).
 Then according to Eq.~(\ref{ConditionForGequivalence}) and property \ref{multi} of characteristic functions (listed in the section \ref{properties}), the characteristic function of $|\psi\rangle\otimes|\nu_{0}\rangle$ is the same as the characteristic function of $|\phi\rangle\otimes|\nu\rangle$. It follows from Theorem \ref{theorem of characteristic} that there exists a G-invariant unitary which maps  $|\psi\rangle\otimes|\nu_{0}\rangle$ to $|\phi\rangle\otimes|\nu\rangle$. So by coupling the system to an environment in state $|\nu_{0}\rangle$ via this G-invariant unitary, and then discarding the environment we can transform $|\psi\rangle$ to $|\phi\rangle$. Note that such a transformation is clearly a G-covariant operation. (Alternatively, let $|\nu^{*}\rangle$ be the state with characteristic function $e^{-i\Theta(g)}$. Note that since  $e^{-i\Theta(g)}$ is also a 1-d representation of the group then by property \ref{Bochner} there exists a state $|\nu^{*}\rangle$ whose characteristic function is  $e^{-i\Theta(g)}$. Then since $|\psi\rangle\otimes|\nu^{*}\rangle$, and $|\phi\rangle\otimes |\nu_{0}\rangle$ have the same characteristic function, by Theorem \ref{theorem of characteristic} there exists a G-invariant unitary which transforms one to the other. Because $|\nu^{*}\rangle$ is a G-invariant state  and because the unitary is G-invariant, the overall operation is G-covariant.)

Using an analogous argument, we can easily deduce that there also exists a G-covariant operation which maps $|\phi\rangle$ to $|\psi\rangle$. Therefore $|\psi\rangle$ and  $|\phi\rangle$ are G-equivalent.

We now prove the other direction of the theorem, that if $|\psi\rangle$ and $|\phi\rangle$ are G-equivalent, then Eq.~(\ref{ConditionForGequivalence}) follows.  By assumption, there exists a G-covariant operation from $|\psi\rangle$ to $|\phi\rangle$ and vice-versa. It then follows from the Stinespring dilation theorem that there exists a G-invariant unitary $V$ and a G-invariant pure state $\left\vert \eta_{1}\right\rangle $ such that
\begin{equation} \label{Oneway}
V|\psi\rangle|\eta_{1}\rangle=|\phi\rangle|\eta_{2}\rangle
\end{equation}
for some pure state $|\eta_{2}\rangle$, and there exists a G-invariant unitary $V^{\prime}$ and a G-invariant pure state $|\eta_{1}^{\prime}\rangle$ such that
$$V^{\prime}|\phi\rangle|\eta_{1}^{\prime}\rangle=|\psi\rangle|\eta
_{2}^{\prime}\rangle$$
for some pure state $|\eta_{2}^{\prime}\rangle$. These
two equations together imply that
\begin{equation} \label{cycle}
V^{\prime}V|\psi\rangle|\eta_{1}\rangle|\eta_{1}^{\prime}\rangle=|\psi\rangle|\eta_{2}\rangle|\eta_{2}^{\prime}\rangle
\end{equation}
Since
$V^{\prime}$ and $V$ are both G-invariant we can deduce that the characteristic functions of $|\psi\rangle|\eta_{1}\rangle|\eta_{1}^{\prime}\rangle$ and $|\psi\rangle|\eta_{2}\rangle|\eta_{2}^{\prime}\rangle$ are equal. i.e.
\begin{equation} \label{cycle4}
\chi_\psi \chi_{\eta_{1}}\chi_{\eta^{\prime}_{{1}}}=\chi_\psi \chi_{\eta_{2}}\chi_{\eta^{\prime}_{{2}}}
\end{equation}
Since $|\eta_{1}\rangle$ and $|\eta_{1}^{\prime}\rangle$ are both G-invariant states the amplitudes of their characteristic functions are always one and so
\begin{equation} \label{equaliy-carnot}
|\chi_\psi|=|\chi_\psi| |\chi_{\eta_{2}}\chi_{\eta^{\prime}_{{2}}}|
\end{equation}
Now suppose $G$ is a connected compact Lie group. Then for any state $\psi$ in a finite-dimensional Hilbert space carrying a projective unitary representation of $G$, $|\chi_\psi|$  is 1 at the identity and is non-vanishing for a neighbourhood around the identity in any direction.  This implies that $|\chi_{\eta_{2}}\chi_{\eta^{\prime}_{{2}}}|$ has value 1 for a neighbourhood around the identity in any direction. By the analyticity, over the group, of the characteristic functions induced by vectors in a finite-dimensional Hilbert space, this implies that $|\chi_{\eta_{2}}\chi_{\eta^{\prime}_{{2}}}|$ is 1 everywhere. \color{black} Therefore  $|\eta_{2}\rangle|\eta_{2}^{\prime}\rangle$ is an invariant state.   Note that it is
this step of the proof which necessitates the restriction to connected compact Lie groups.

Since $|\eta_{2}\rangle|\eta_{2}^{\prime}\rangle$ is G-invariant then $|\eta_{2}\rangle$ is also G-invariant. Therefore Eq.~(\ref{Oneway}) implies that
\begin{equation}
\chi_{\psi}(g)=\chi_{\phi}(g)e^{i[\Theta_{2}(g)-\Theta_{1}(g)]}
\end{equation}
where $e^{i\Theta_{1}(g)}$ and $e^{i\Theta_{2}(g)}$ are respectively the characteristic functions of $|\eta_{1}\rangle$ and $|\eta_{2}\rangle$.  Finally, because $e^{i\Theta_{1}(g)}$ and $e^{i\Theta_{2}(g)}$ are 1-dimensional representations of G, it follows that $e^{i[\Theta_{2}(g)-\Theta_{1}(g)]}$ is as well. This completes the proof of Theorem \ref{maintheorem}.

As we mentioned above, there is only one point in the proof in which we use the assumption that the group is a connected Lie group: the fact that 
$|\chi_\psi|=|\chi_\psi| |\chi_{\eta_{2}}\chi_{\eta^{\prime}_{{2}}}|$ implies   $|\chi_{\eta_{2}}\chi_{\eta^{\prime}_{{2}}}|=1$. This follows from the analyticity of the characteristic functions for finite-dimensional representations of Lie groups.  For finite groups, where we cannot appeal to analyticity, if $|\chi_\psi|$ is zero at some $g\in G$ then $|\chi_\psi|=|\chi_\psi| |\chi_{\eta_{2}}\chi_{\eta^{\prime}_{{2}}}|$ does not imply $|\chi_{\eta_{2}}\chi_{\eta^{\prime}_{{2}}}|=1$ at that point.  However, if we assume the function $\chi_{\psi}$ is nonzero for all $g\in G$ then we can again deduce $|\chi_{\eta_{2}}\chi_{\eta^{\prime}_{{2}}}|=1$ and the rest of the argument goes through as before. This completes the proof of theorem \ref{general-deterministic}.
\end{proof}

\begin{example}
Recall our quantum optics example where the set of all phase shifts forms a representation of group U(1) (see example \ref{Ex-U(1)-reduc}). For this representation of the symmetry U(1) it turns out that  the criterion of U(1)-equivalence of pure states has a simple form in terms of reductions onto irreps. Suppose that the probability distributions over integers $p_\psi$ and $p_\phi$ are the reductions onto the irreps of $\psi$ and $\phi$ respectively, so that the characteristic functions are the Fourier transforms of these.  Theorem \ref{maintheorem} implies that $\psi$ and $\phi$ are U(1)-equivalent if and only if there exists an integer $\Delta$ such that $\sum_n p_\psi(n) e^{i n \theta} = e^{i \Delta \theta} \sum_n p_\phi(n) e^{i n \theta}$, or equivalently, using the Fourier transform, such that
\begin{equation}
p_\psi(n)=p_\phi(n+\Delta),
\end{equation}
which is precisely the condition found in Ref.~\cite{GS07}.   As a specific example, we can see that the states $|\psi\rangle = \frac{1}{\sqrt{2}}(|0\rangle+|1\rangle)$ and $|\phi\rangle = \frac{1}{\sqrt{2}}(|2\rangle+|3\rangle)$ are U(1)-equivalent either by noting that $\chi_{\psi}(\theta) = e^{i2\theta}\chi_{\phi}(\theta)$ or by noting that $p_\psi(n)=p_\phi(n-2)$.
\end{example}




In the above proof, free operations $V$ and $V'$ together generate a closed reversible  cycle: we start with state $|\psi\rangle$ (the resource) and use an invariant state  $|\eta_{1}\rangle$ (a non-resource) to generate $|\phi\rangle|\eta'_{1}\rangle$ and then use $|\psi\rangle$ and couple it to   $|\eta_{2}\rangle$ to get the state $|\psi\rangle|\eta'_{2}\rangle$.
Using the properties of characteristic functions,
we showed that the residue states $|\eta_{2}\rangle$ and $|\eta_{2}^{\prime}\rangle$ should be invariant (non-resources). However this property can be derived from more general considerations.   Suppose $|\eta_{2}\rangle|\eta_{2}^{\prime}\rangle$ is not invariant. This implies that by going through this cycle we have generated some additional resource without consuming any. 
This should be impossible if the state $|\psi\rangle$ contains only a finite amount of the resource, which is indeed the case for any state on a finite-dimensional Hilbert space if the group is not finite.




\section{Deterministic transformations} \label{sec:deterministic}

In this section we find the necessary and sufficient condition to determine whether a pure state $\psi$ can be transformed to a pure state $\phi$ by a G-covariant channel. This is distinct from the question of G-equivalence because the transformation is not required to be reversible.

\begin{theorem} \label{deterministic}
There exists a deterministic G-covariant map $\mathcal{E}$
transforming $\psi$ to $\phi$ if and only if there exists a positive definite function $f$ over the group $G$ such that $\chi_{\psi}(g)=\chi_{\phi}(g)f(g)$ for all $g\in G$.
\end{theorem}
Note that if $\chi_{\phi}$ is  nonzero for all $g\in G$ then $f(g)=\chi_{\psi}(g)/\chi_{\phi}(g)$. So, in this case we can conclude that there exists a  G-covariant map $\mathcal{E}$
transforming $\psi$ to $\phi$ if and only if  $\chi_{\psi}(g)/\chi_{\phi}(g)$ is a positive definite function.  As it is discussed in the appendix  \ref{app:qcharfuncs} one can test positive definiteness of $f(g)$ by verifying that the set of operators defining its Fourier transform are all positive. \blk



\begin{proof} (Theorem \ref{deterministic})
As in the proof of theorems \ref{maintheorem} and \ref{general-deterministic},
the main tool we use in this proof is the Stinespring dilation theorem discussed in the preliminaries (See \cite{thesis:Marvian} and \cite{Werner}) . By this result we know that  the transformation can be achieved if and only if one can find an initial invariant ancilla state $\eta$ and a final (possibly non-invariant) ancilla state $\nu$ such that $\psi\otimes\eta$ and $\varphi\otimes\nu$ are unitarily G-equivalent. \ One then discards $\nu$ at the end. \ In terms of characteristic functions, we require
\begin{equation}
\chi_{\psi}(g)e^{i\Theta(g)}=\chi_{\varphi}(g)\chi_{\nu}(g).
\end{equation}
where $e^{i\Theta(g)}$ is a 1-d representation of the group, the characteristic function of the invariant state $\eta$,  and $\chi_{\nu}(g)$ is the characteristic function of the discarded state $\nu$. This implies  $\chi_{\psi}(g)=\chi_{\varphi}(g) \left[\chi_{\nu}(g)e^{-i\Theta(g)}\right]$. Since $\chi_{\nu}(g)$ and $e^{-i\Theta(g)}$ are both positive definite, so is their product (see appendix~\ref{app:charfuncs}). This proves one direction of the theorem.  To prove the other direction, suppose there exists a positive definite function $f(g)$ such that   $\chi_{\psi}(g)=\chi_{\phi}(g)f(g)$ for all $g\in G$. This obviously implies $f(e)=\chi_{\psi}(e)/\chi_{\varphi}(e)=1$ and so the function is normalized. Then according to property \ref{Bochner} of characteristic functions, there exists a normalized state $\nu$ whose characteristic function is equal  to $f(g)$. Now because the characteristic function of $\phi\otimes\nu$, i.e. $\chi_{\varphi}(g)f(g)$, is equal to $\chi_{\psi}$, they are unitarily G-equivalent. Therefore, there exists a G-invariant unitary transforming $\psi\otimes \nu_{0}$ to $\phi\otimes\nu$ where $\nu_{0}$ is the G-invariant state whose  characteristic function is constant and equal to one for all group elements. So by applying this G-invariant unitary to $\psi\otimes \nu_{0}$ and transforming it to $\phi\otimes\nu$ and then discarding $\nu$ we can transform $\psi$ to $\phi$. Obviously this transformation is G-covariant.\end{proof}

It is worth noting that the necessary and sufficient condition for G-equivalence (Theorems \ref{maintheorem} and \ref{general-deterministic}) can also be obtained from the above result on deterministic transformations: If $\psi$ and $\phi$ are G-equivalent, then there exist a G-covariant transformation from $\psi$ to $\phi$ and a G-covariant transformation from $\phi$ to $\psi$. Then the above results imply that there exist normalized positive definite functions $f_{1}$ and $f_{2}$ such that $\chi_{\psi}(g)=\chi_{\phi}(g)f_{1}(g)$ and $\chi_{\phi}(g)=\chi_{\psi}(g)f_{2}(g)$. Substituting the second equation into the first, we have
\begin{equation}\label{defneffs}
\chi_{\psi}(g)=\chi_{\psi}(g)f_{1}(g) f_{2}(g),
\end{equation}
and so if $\chi_{\psi}(g)$ is nonzero for all group elements  it follows that $\forall g\in G: f_{1}(g) f_{2}(g)=1$.  Given that $\forall g\in G: |f_{1}(g)|,|f_{2}(g)|\le 1$ (because the absolute value of a positive definite function at any $g$ is always less than or equal to its  absolute value at $e$ and $f_{1}(e),f_{2}(e)=1$ by virtue of Eq.~\ref{defneffs}), we infer that $\forall g\in G: |f_{1}(g)|,|f_{2}(g)| = 1$. It follows therefore that $f_{1}$ and $f_{2}$ are 1-dimensional representations of the group, which is the content of theorem~\ref{general-deterministic}. One can prove theorem \ref{maintheorem} similarly for the case of connected compact Lie groups.

In the following we present  two examples, corresponding to the groups $U(1)$ and $Z_{N}$.



\subsection{ Example: U(1)-covariant deterministic transformations }

Recall our quantum optics example where the set of all phase shifts forms a representation of group U(1) (see example \ref{Ex-U(1)-reduc}). According to the  theorem \ref{deterministic}   there exists a deterministic U(1)-covariant map transforming $\psi$ to $\phi$ if and only if there exists a positive definite  $f({\theta})$ such that
\begin{equation}\label{charfnspsiphi}
\chi_{\psi}(\theta)=f({\theta})\chi_{\phi}({\theta})
\end{equation}
Since $f(\theta)$ is positive definite all Fourier components of this function $\{q_n\}$ are positive. Furthermore, since $\chi_{\psi}(0)=\chi_{\phi}(0)=1$ we conclude that $f(0)=1$ which implies that $\sum_{n} q_{n}=1$ and so the set $\{q_n\}$ is also a probability distribution. Suppose the  probability distributions over integers $p_{\psi}$ and $p_{\phi}$ are the Fourier transforms of $\chi_{\psi}$ and $\chi_{\phi}$ respectively. Then the Fourier transform of Eq.~(\ref{charfnspsiphi}) yields
\begin{equation}
p_{\psi}(n)=\sum_k p_{{\phi}}({n-k}) q(k).
\end{equation}
So the U(1)-covariant transformation from $\psi$ to $\phi$ exists iff there exists a probability distribution $q$ over integers which satisfies the above equality. This is indeed the condition for deterministic interconversion in the U(1) case found in Ref. \cite{GS07}.

\subsection{Example: $Z_N$-covariant deterministic transformations }

Suppose the group under consideration is the  group $Z_N$, the cyclic group of order $N$. For any $N$, the group $Z_N$ is isomorphic to the group of integers $\{0,...,N-1\}$ where the group action is addition modulo $N$. We use this isomorphism to denote the group elements.  These groups are clearly Abelian and so all of their irreps are one-dimensional.  We can easily see that these irreps can be identified by an integer $J$ in the set $ \{0,...N-1\}$ such that the irrep labeled by $J$ is
\begin{equation}
k\in Z_{N}\rightarrow U_{J}(k)=e^{i2\pi Jk/N}.
\end{equation}
So an arbitrary (non-projective) unitary representation of $Z_N$,  $k\in Z_{N}\rightarrow U(k)$,  can be decomposed as
\begin{equation}
U(k)=\bigoplus_{J,\alpha} e^{i Jk 2\pi/N } |J,\alpha\rangle\langle J,\alpha|,
\end{equation}
where $\alpha$ labels copies of irrep $J$ and  $\{|J,\alpha\rangle\}$ is a basis for the Hilbert space. An arbitrary state $\psi$ in this basis can be expanded as
\begin{equation}
|\psi\rangle=\sum _{J,\alpha} \psi(J,\alpha)  |J,\alpha\rangle.
\end{equation}
As with any other Abelian group, the reduction of the state onto the irreps is simply the probability distribution that the state induces over the irreps. So the reduction of $\psi$ is specified by the probability distribution
$$\{p_\psi(J)\equiv \sum_\alpha |\psi(J,\alpha)|^2: J=0,...,N\}.$$
On the other hand, the  characteristic function of $\psi$ is by definition the  function $k\in\{0\cdots N-1\}\rightarrow \langle\psi|U(k)|\psi\rangle$, that is,
\begin{equation}
\chi_\psi(k)=\sum _{J,\alpha} |\psi(J,\alpha)|^2 e^{i2\pi Jk/N}.
\end{equation}
Clearly the characteristic function is the discrete Fourier transform of the reduction of the state onto the irreps.

Now we are interested to know whether there exists a $Z_N$-covariant quantum operation which transforms $\psi$ to $\phi$.
Assuming the  characteristic function of $\phi$, $\chi_\phi(k)$, is nonzero for all $k$'s, it follows from theorem \ref{deterministic} that such a $Z_N$-covariant map exists iff $\chi_\psi(k)/\chi_\phi(k)$ is a positive definite function. But this function is positive definite iff its Fourier transform is always positive, i.e. iff
\begin{equation}
q(J)\equiv \sum_{k} \frac{ \chi_\psi(k)}{\chi_\phi(k)} e^{i2\pi Jk/N},
\end{equation}
is  positive for all $J=0,...,N$. So to summarize, the necessary and sufficient condition for the existence of a $Z_{N}$-covariant channel which transforms $\psi$ to $\phi$ is that
\begin{equation}\label{ZN-cond}
\forall J\in\{0,...,N\}: \ \ \  \sum_{k} \frac{ \chi_\psi(k)}{\chi_\phi(k)} e^{i2\pi Jk/N}\ge 0
\end{equation}

Consider the case of $Z_{2}$ which has only two group elements denoted by $\{e,\pi\}$ where $e$ is the identity of the group and $\pi^{2}=e$. Using the above convention we denote $e$ by $k=0$ and $\pi$ by $k=1$.  This group has only two inequivalent irreps: The trivial representation ($J=0$) in which $$U_{J=0}(0)=U_{J=0}(1)=1$$ and the nontrivial ($J=1$) in which $$U_{J=1}(1)=-U_{J=1}(0)=-1.$$  Then the reduction of $\psi$ onto irreps is specified by the probability assigned to each of these irreps and because there are only two irreps we only need to specify one of the probabilities, say $p_\psi(J=0)$. The characteristic function of $\psi$ is
\begin{equation}
\chi_{\psi}(k)= p_\psi(J=0)+(-1)^{k} p_\psi(J=1)
\end{equation}
So $\chi_{\psi}(0)=1$ and  $\chi_{\psi}(1)=2 p_\psi(J=0)-1$.
Then Eq.~(\ref{ZN-cond}) implies that the transformation $\psi\xrightarrow{\textnormal{G-cov}} \phi$ is possible iff
\begin{equation}
q(0)= \frac{ \chi_\psi(0)}{\chi_\phi(0)}+ \frac{ \chi_\psi(1)}{\chi_\phi(1)} \ge 0
\end{equation}
and
\begin{equation}
q(1)= \frac{ \chi_\psi(0)}{\chi_\phi(0)}- \frac{ \chi_\psi(1)}{\chi_\phi(1)} \ge 0.
\end{equation}
Since $\frac{ \chi_\psi(0)}{\chi_\phi(0)}$ is always equal to one it turns out that the above two inequalities are equivalent to $|\chi_{\psi}(1)|\le |\chi_{\phi}(1)|$,  i.e.
\begin{equation}
\left| p_{\psi}(J=0)-p_{\psi}(J=1)  \right| \le \left| p_{\phi}(J=0)-p_{\phi}(J=1)  \right|
\end{equation}
Since $$p_{\psi}(J=0)+p_{\psi}(J=1)=p_{\phi}(J=0)+p_{\phi}(J=1)=1 $$
the above condition is equivalent to the condition
\begin{equation}
\min\{p_{\phi}(J=0),p_{\phi}(J=1)\} \le \min\{p_{\psi}(J=0),p_{\psi}(J=1)\}
\end{equation}
which is exactly the same condition previously obtained  in \cite{GS07} using a totally different approach. Eq. (\ref{ZN-cond}) is the generalization of this specific result for arbitrary cyclic group $Z_{N}$.

\section{Catalysis}\label{sec:catalysis}

In any resource theory, if state $\psi$ cannot be converted to state $\phi$ deterministically under the restricted operations, it may still be the case that it is possible to do so using a \emph{catalyst}, which is an ancillary system that is prepared in a state that is \emph{not} free relative to the restriction that defines the resource theory but which must
be returned to its initial state at the end of the procedure.  For example, in the resource theory of entanglement it is a well-known fact that a transformation from a given state  to another might be forbidden under LOCC but  that transformation can be performed using LOCC and an appropriate catalyst \cite{Jon-Pel}.

In the case of the resource theory of asymmetry, a catalyst is a finite-dimensional ancillary system in an \emph{asymmetric} state which can be used to achieve the interconversion but only in such a way that its state remains unchanged at the end of the  process.

We shall say that the conversion $\psi$ to $\phi$ is a nontrivial example of catalysis if there is no deterministic G-covariant channel  under which $\psi$ goes to $\phi$, but there is a deterministic G-covariant channel and a catalyzing state $\zeta$ such that $\psi\otimes \zeta$ goes to $\phi \otimes \zeta$.

In the resource theory of asymmetry, whether there is a nontrivial catalysis or not depends on the nature of the group. In the following we  prove that in the case of compact connected Lie groups, catalysts are totally useless.  We also present an example which shows how catalysts can be useful in the case of finite groups.

It turns out that in the case of pure state transformations, characteristic functions give us a powerful insight into how a catalyst can make a transformation possible.  Assume $\chi_{\psi}$ and $\chi_{\phi}$ are respectively the characteristic functions of  states $\psi$ and $\phi$  for which there is no G-covariant transformation which takes $\psi$ to $\phi$.  Then from theorem \ref{deterministic} we know that if there is no G-covariant transformation from $\psi$ to $\phi$ then there is no analytic positive definite function $f$ over the group G that satisfies
\begin{equation}
\forall g\in G:\  \chi_{\psi}(g)=\chi_{\phi}(g)f(g)
\end{equation}
 On the other hand, if this transformation is possible using a catalyst $\zeta$ with characteristic function $\chi_{\zeta}$ then there should exist an analytic positive definite function $f'$ such that
\begin{equation}\label{catal}
\forall g\in G: \chi_{\psi}(g)\chi_{\zeta}(g)=\chi_{\phi}(g)\chi_{\zeta}(g)f'(g).
\end{equation}
Now clearly for all points $g\in G$ for which $\chi_{\zeta}(g)\neq 0$, Eq.(\ref{catal})  implies
$\chi_{\psi}(g)=\chi_{\phi}(g)f'(g)$. But we know that this equality cannot hold for all group elements, otherwise there exists a G-covariant channel which transforms $\psi$ to $\phi$, in contradiction with our assumption. This argument shows that the role of a catalyst is specified by the elements of the group at which the characteristic function of the catalyst is zero; For these specific group elements, although  $\chi_{\psi}(g)\neq\chi_{\phi}(g)f'(g)$, nonetheless $\chi_{\psi}(g)\chi_{\zeta}(g)=\chi_{\phi}(g)\chi_{\zeta}(g)f'(g)$.  This argument shows that there is an important distinction   between  the cases of compact connected Lie groups and finite groups or Lie groups which are not connected.

\subsection{Compact connected Lie groups}

In the case of compact connected Lie groups, using the above argument and by virtue of  the analyticity of  characteristic functions one can argue that catalysts cannot help, i.e. if a transformation is possible with a catalyst, it is also possible without any catalyst. To see this, first note that for any finite-dimensional representation of a compact Lie group there is a neighborhood around  the identity element of the group within which the characteristic functions of all pure states are nonzero (Otherwise there would be a unitary which is arbitrarily close to identity for which $\langle\psi|U|\psi\rangle=0$ for some state $\psi$, but in a finite-dimensional Hilbert space this is not possible.). This implies that in this neighborhood, if Eq.~(\ref{catal}) holds then the following equation holds
\begin{equation}
 \chi_{\psi}(g)=\chi_{\phi}(g)f'(g)
\end{equation}
But since all these functions are analytic and since the group G is connected,  if the above equality is true for a neighbourhood around the identity element of G then it will be true for all G. Then by theorem   \ref{deterministic} we can conclude that there exists a G-covariant channel which transforms $\psi$ to $\phi$ (without the help of any catalyst). So if this transformation is possible with the use of a catalyst then it is also possible without using the catalyst. So to summarize we have proven  that
\begin{theorem}\label{catalysisLie}
For symmetries associated with compact connected Lie groups, there are no examples of nontrivial catalysis using a finite catalyst.
\end{theorem}






\subsection{Finite groups}
The above argument clearly does not work in the case of finite groups. Indeed, as we will see in the following, in the case of finite groups there are states for which the characteristic function is  zero for all $g\in G$ except the identity. If we use such a state as a catalyst, Eq.~(\ref{catal}) holds for all group elements and consequently for any pair of states $\psi$ and $\phi$, one can always transform one to the other using the catalyst. (Indeed as we show in the following one can always transform any mixed state to any other mixed state using such a catalyst.)

For a group with a finite number of elements, it is possible for the catalyst to consist of a system with a Hilbert space $\mathcal{H}$ having dimension greater than or equal to the order of the group. In this case, the representation of the group can be the left regular representation on the Hilbert space $\mathcal{H}$, $g\rightarrow T_{L}(g)$, such that
\begin{equation}
\forall g\in G:\ \ \ T_{L}(g)|h\rangle=|gh\rangle,
\end{equation}
where $\{|h\rangle:h\in G\}$ is an orthonormal basis for $\mathcal{H}$.
Now note that the characteristic function of any state $|h\rangle$ is $\chi_h(g)=\langle h|T_{L}(g)|h\rangle=\delta_{e,g}$, the Kronecker-delta function centered on the identity group element.  Eq.~(\ref{catal}) then implies that such a state can catalyze any pure-to-pure transformation.

 Also it is straightforward to show that for any pair of states $\rho$ and $\sigma$ (pure or mixed) there exists a G-covariant channel which transforms $ \rho\otimes|h\rangle\langle h|$ to $\sigma\otimes|h\rangle\langle h|$.  
 One realization of this G-covariant map is the following
\begin{equation}
\mathcal{E}_h(X) \equiv \sum_{g\in G}\textnormal{tr}\left( \left[\mathbb{I}\otimes|g\rangle\langle g|\right] X  \right)  U(gh^{-1})\sigma U^{\dag}(gh^{-1})\otimes|g\rangle\langle g| ,
\end{equation}
where $g\rightarrow U(g)$ is the representation of the symmetry on the space where $\sigma$ lives and $\mathbb{I}$ is the identity operator acting on the Hilbert space of $\rho$ \footnote{This fact can be be made intuitive by imagining that the restriction to G-covariant operations results from one party, Bob, lacking a shared reference frame with another party, Alice.  In this case, our interconvertibility problem is described as follows: Alice sends to Bob a pair of systems which are described by the state $\rho\otimes |h\rangle\langle h|$ relative to her frame and asks him to transform these to $\sigma \otimes |h\rangle\langle h|$, again relative to her frame. One can think of the second system as a token of Alice's reference frame.  Because the group is finite, Bob can simply measure $\{ |g\rangle \langle g| : g\in G\}$ on the token and determine precisely the relationship between their reference frames.  Thereafter, he can perform any operation relative to Alice's frame.  In other words, for finite groups, the fact that one can prepare a perfect token of a reference frame using a finite-dimensional system is equivalent to the fact that one can find a finite-dimensional catalyst that makes possible any state transformation.}.

So unlike the case of connected compact Lie groups, in the case of a symmetry described by a finite group, catalysts can be helpful.

\section{state-to-ensemble and stochastic transformations} \label{sec:stochastic}

In this section, we study the problem of transforming one pure state to an ensemble of pure states using G-covariant operations. We are interested to know whether  it is possible to transform a given state $\psi$  to the state $\phi_i, i=1,...,N$ with probability $p_i$. The transformation is such that at the end we know $i$ and so we know which $\phi_i$ is generated.

\begin{theorem}\label{state_ensemble}
There exists a G-covariant map transforming $\psi$ to $\left\{  (p_{i},\left\vert \phi_{i}\right\rangle) \right\} $ if
and only if there exists  positive-definite (and continuous when $G$ is a Lie group)  functions $f_{i}(g)$ for which $f_{i}(e)=1$ such that
\begin{equation}\label{cond-theorem-Stoch}
\chi_{\psi}(g)=\sum_{i}p_{i} f_{i}(g)\chi_{\phi_{i}}(g).
\end{equation}
\end{theorem}
One important special case is when we are interested in just one of the outcome states. In particular we are interested to know whether we can transform state $|\psi\rangle$ to $|\phi\rangle$ with probability $p$. We call these transformations \emph{stochastic transformations}. The above theorem implies the following corollary about stochastic transformations.
\begin{corollary}\label{stochastic}
There exists a G-covariant map  taking $\psi$ to $ \phi$ with probability $p$ iff there exists a positive definite (and continuous when $G$ is a Lie group)  function $f(g)$ for which $f(e)=1$ such that $\chi_{\psi}(g)-p\chi_{\phi}(g)f(g)$ is positive definite.
\end{corollary}
These results are proven at the end of the section.

\subsection{Example: U(1)-covariant stochastic maps}

Recall our quantum optics example where the set of all phase shifts forms a representation of group U(1) (see example \ref{Ex-U(1)-reduc}).  Let $\textnormal{Irreps}_{\textnormal{U}(1)}(\psi)$ be the set of eigenvalues  of the number operator ${N}$ to which the pure state $\psi$ assigns nonzero weight. Assuming that $\psi$ can be transformed to $\phi$  with nonzero probability under a U(1)-covariant operation, one can easily show that
\begin{enumerate}
\item
The cardinality of  $\textnormal{Irreps}_{\textnormal{U}(1)}(\psi)$  is larger than or equal to the cardinality of  $\textnormal{Irreps}_{\textnormal{U}(1)}(\phi)$, i.e.,
\begin{equation}
\left|\textnormal{Irreps}_{\textnormal{U}(1)}(\phi)\right|\le \left|\textnormal{Irreps}_{\textnormal{U}(1)}(\psi)\right|
\end{equation}
\item
\begin{align*}
\textnormal{max}\{\textnormal{Irreps}_{\textnormal{U}(1)}(\phi) \}-\textnormal{min}\{\textnormal{Irreps}_{\textnormal{U}(1)}(\phi) \}\le\\  \textnormal{max}\{\textnormal{Irreps}_{\textnormal{U}(1)}(\psi) \}-\textnormal{min}\{\textnormal{Irreps}_{\textnormal{U}(1)}(\psi) \}
\end{align*}

\end{enumerate}


Here, we prove item 2 by contradiction.  Assume this condition does not hold. Then for any positive definite function $f(\theta)$,  $\chi_{\phi}(\theta)f(\theta)$ has a nonzero component of $e^{im\theta}$   for some $m$ such that $m<n_{\textnormal{min}}(\psi)$ or $m>n_{\textnormal{max}}(\psi)$. Since both $\chi_{\phi}(\theta)$ and $f(\theta)$ are positive-definite, the coefficient of $e^{im\theta}$  will be positive. This implies that for any nonzero probability $p$, the coefficient of $e^{im\theta}$ in $\chi_{\psi}(\theta)-p\chi_{\phi}(\theta)f(\theta)$ is negative and so the function $\chi_{\psi}(\theta)-p\chi_{\phi}(\theta)f(\theta)$ is not positive definite for any nonzero $p$. This proves the claim. Item 1 is proven similarly.  

Item 2 was obtained by a different argument in Ref.~\cite{GS07}.\footnote{The set $\textnormal{Irreps}_{\textnormal{U}(1)}(\phi)$ was called the ``number spectrum'' of $\phi$ in Ref.~\cite{GS07}.} \blk

\subsection{Example: SO(3)-covariant stochastic maps}

Let $\textnormal{Irreps}_{\textnormal{SO}(3)}(\psi)$ be the set of all angular momentum quantum numbers $j$ corresponding to irreps  of SO(3) to which the pure state $\psi$ assigns nonzero weight.

Using a similar argument to the one we used for the case of U(1), one can easily conclude that if $\psi$ can be transformed to $\phi$ under an SO(3)-covariant channel, then
\begin{enumerate}
\item
the cardinality of  $\textnormal{Irreps}_{\textnormal{SO}(3)}(\psi)$  is larger than or equal to the cardinality of  $\textnormal{Irreps}_{\textnormal{SO}(3)}(\phi)$, i.e.,
\begin{equation}
\left|\textnormal{Irreps}_{\textnormal{SO}(3)}(\phi)\right|\le \left|\textnormal{Irreps}_{\textnormal{SO}(3)}(\psi)\right|.
\end{equation}

\item
\begin{align*}
\textnormal{max}\{\textnormal{Irreps}_{\textnormal{SO}(3)}(\phi) \}-\textnormal{min}\{\textnormal{Irreps}_{\textnormal{SO}(3)}(\phi) \}\le\\ \textnormal{max}\{\textnormal{Irreps}_{\textnormal{SO}(3)}(\psi) \}-\textnormal{min}\{\textnormal{Irreps}_{\textnormal{SO}(3)}(\psi) \}
\end{align*}
\item
\begin{equation}
\textnormal{max}\{\textnormal{Irreps}_{\textnormal{SO}(3)}(\phi) \}\le \textnormal{max}\{\textnormal{Irreps}_{\textnormal{SO}(3)}(\psi) \}
\end{equation}

\end{enumerate}

The proofs of items 1 and 2 are similar to the case of U(1). To prove item 3 note that the maximum value of $j$ to which $\chi_{\phi}(\theta)f(\theta)$ assigns nonzero weight is greater than or equal to $j_\textnormal{max}({\phi})$. So if $j_\textnormal{max}({\phi})$ is strictly greater than $j_\textnormal{max}({\psi})$, then for any nonzero $p$,  $\chi(\psi)-p\chi_{\phi}(\theta)f(\theta)$ cannot be positive definite.

Item 3 implies that if a pure state does not have any component of angular momentum higher than $j$ then  by  rotationally covariant operations it cannot be transformed with nonzero probability to another pure state which assigns some amplitude to an angular momentum higher than $j$.


\subsection{Proof of theorem \ref{state_ensemble}}
According to a version of the Stinespring dilation theorem, a general state-to-ensemble transformation  can always be purified in the following way: First, the input system (with Hilbert space $\mathcal{H}_{\textnormal{in}}$)  unitarily interacts with an ancillary system (with Hilbert space $\mathcal{H}_{\textnormal{anc}}$). Now we consider the total Hilbert space $\mathcal{H}_{\textnormal{in}}\otimes\mathcal{H}_{\textnormal{anc}}$ as
\begin{equation}
\mathcal{H}_{\textnormal{in}}\otimes\mathcal{H}_{\textnormal{anc}}=\bigoplus_i \mathcal{H}_{i}\otimes \mathcal{H}_{i}' \otimes |i\rangle\langle i|
\end{equation}
After the unitary time evolution we perform a projective measurement on the third subsystem in the basis $\{|i\rangle\langle i|\}$ and according to the outcome of measurement we discard the subsystem  $ \mathcal{H'}_{i}$.  The output would be the system described by $\mathcal{H}_{i}$. This procedure realizes the most general state-to-ensemble transformation.

Suppose a transformation maps $|\psi\rangle$ to $|\phi_{i}\rangle$ with  probability $p_{i}$. Since the output is pure, clearly it cannot  be entangled with the discarded system. In other words, after applying the unitary $V$ which couples the system and ancilla  the total state should be in the form
\begin{equation}
V|\psi\rangle|\nu\rangle=\sum_i \sqrt{p_i} |\phi_i\rangle |\eta_i\rangle |i\rangle
\end{equation}
where $|\psi\rangle$ is the initial state of the system and $|\nu\rangle$ is the initial state of the ancilla.

Now  according to an extension of Stinespring's dilation theorem for G-covariant quantum operations,  if  the state-to-ensemble transformation is G-covariant then one can choose  the initial state $|\nu\rangle$ of ancilla,  the unitary $V$, and the basis $\{|i\rangle\}$ to all be  G-invariant \cite{Werner} .

Assuming $V$ is a G-invariant unitary then the characteristic function of the right hand side should be equal to the characteristic function of $|\psi\rangle|\nu\rangle$.
This implies
\begin{equation}
\chi_{\psi}(g)e^{i\theta(g)}=\sum_{i}p_{i}\ \chi_{\nu_{i}}(g)\chi_{\phi_{i}}(g) e^{i\alpha_i(g)}
\end{equation}
where $e^{i\theta(g)}$ is the characteristic function of the G-invariant ancilla $|\nu\rangle$ and $\{e^{i\alpha_i(g)}\}$ are the characteristic functions of the G-invariant states $\{|i\rangle\}$. Now because the product of two characteristic functions is also a characteristic function,  $\chi_{\nu_{i}}(g)e^{i\alpha_i(g)}e^{-i\theta(g)}$ is a valid characteristic function. So if there exists a G-covariant transformation which maps state $\psi$ to $\phi_{i}$ with probability $p_{i}$, then the equation \eqref{cond-theorem-Stoch} should hold.
This completes the proof of one direction of the theorem. To prove the other direction, we note that property \eqref{Bochner} of characteristic functions  listed in section \ref{properties} implies that there exists a set of states $\{|\nu_{i}\rangle\}$ which have characteristic functions equal to $\{f_{i}\}$. Now we choose $|\nu\rangle$, the initial state of the ancilla, to be a G-invariant state and we assume that its characteristic function is equal to 1 for all group elements (i.e. any group element maps $|\nu\rangle$ exactly to itself).  Similarly we choose a basis $\{|i\rangle\}$ to be a set of G-invariant orthonormal states and assume the characteristic functions of all of them are constant and equal to 1. Then, equation \eqref{cond-theorem-Stoch} implies that the characteristic function of $|\psi\rangle|\nu\rangle$ is equal to the characteristic function of $\sum_i \sqrt{p_i} |\phi_i\rangle |\eta_i\rangle |i\rangle$ and so there exists a G-invariant unitary which maps the former state to the latter. Now by performing a measurement in the basis $\{|i\rangle\}$ and discarding the subsystem with the state $ |\eta_i\rangle$ we can realize the desired map. This completes the proof.

\section{Acknowledgement}
We thank Sarah Croke for a discussion about Gram matrices, Giulio Chiribella for a discussion about Noether's theorem and Gilad Gour for general discussions.  Perimeter Institute is supported by the Government of Canada through Industry Canada and by the Province of Ontario through the Ministry of Research and Innovation. I. M. is supported by a Mike and Ophelia Lazaridis fellowship.

\appendix


\section{Short review of projective unitary representations} \label{app:proj}

In this section we list some useful definitions and properties of projective unitary representations of groups.

Two projective unitary representations of a group, $g\rightarrow U(g)$ acting on space $\mathcal{H}$ and $g\rightarrow V(g)$ acting on space $\mathcal{K}$,  are \emph{equivalent} iff  there exists an isometry  $T:\mathcal{H}\rightarrow\mathcal{K}$ such that $TT^{\dag}=\mathbb{I}_{\mathcal{K}}$ and  $T^{\dag}T=\mathbb{I}_{\mathcal{H}}$, where $\mathbb{I}_{\mathcal{K}}$ and $\mathbb{I}_{\mathcal{H}}$ are the identity operators on  $\mathcal{K}$ and $\mathcal{H}$ respectively,  and $\forall g\in G:\ TU(g)T^{\dag}=V(g)$.

Consider an arbitrary projective unitary representation of a group on a space. We say a subspace of this space is \emph{invariant} under the action of a group, if  under the action of any arbitrary element of the group any vector in the subspace is mapped to a vector in this subspace.

A representation on a space is called an \emph{irreducible} representation (\emph{irrep} for short) if there is no proper subspace of the space (i.e. a nonzero subspace which is not equal to the total space) which remains invariant under the action of the group.  Equivalent irreps can be grouped in the same equivalence class,  labeled by the Greek index $\mu$.

Note that the unitarity of a projective unitary representation implies that all the irreps which show up in that representation should have the same cocycle.  Any two projective unitary representations $g\rightarrow U(g)$ and  $g\rightarrow V(g)$ which  have the same cocyle, i.e. $U(g_{1})U(g_{2})=\omega(g_{1},g_{2})U(g_{1}g_{2})$ and   $V(g_{1})V(g_{2})=\omega(g_{1},g_{2})V(g_{1}g_{2})$ for a cocycle $\omega(g_{1},g_{2})$ are said to be in the same \emph{factor system}.

\begin{theorem}
Any projective unitary representation of a finite or a compact Lie group can be decomposed into a direct sum of a discrete number of finite-dimensional projective unitary irreps which are all in the same factor system.
\end{theorem}

Suppose $\{U(g):g\in G \}$ is a projective unitary representation of a finite or compact Lie group $G$ on the Hilbert space $\mathcal{H}$.  Then, the  decomposition of this representation to irreps suggests the following decomposition of the Hilbert space
\begin{equation}
\mathcal{H}=\bigoplus_{\mu} \mathcal{M}_{\mu}\otimes  \mathcal{N}_{\mu},
\end{equation}
where $\mu$ labels inequivalent unitary projective irreps in the same factor system,  $\mathcal{M}_{\mu}$ is the subsystem on which $\{U(g):g\in G \}$  acts like irrep $\mu$ of G  and $\mathcal{N}_{\mu}$ is the subsystem associated to the copies of representation $\mu$ (the dimension of $ \mathcal{N}_{\mu}$ is equal to the multiplicity of the irrep $\mu$ in this representation). Then  $U(g)$ can be written as
\begin{equation}  \label{decomposition}
U(g) =\bigoplus_\mu U_\mu(g) \otimes \mathbb{I}_{\mathcal{N}_\mu}
\end{equation}
where $U_\mu(g)$ acts on $\mathcal{M}_{\mu}$ irreducibly and where $\mathbb{I}_{\mathcal{N}_\mu}$ is the identity operator on the multiplicity subsystem $\mathcal{N}_{\mu}$.

Now by Schur's lemmas it follows  that any operator $A$ which commutes with all unitaries $\{U(g): g\in G\}$ should be in the following form
\begin{equation} \label{G-inv operators}
A=\bigoplus_\mu \mathbb{I}_{\mathcal{M}_\mu} \otimes A_{\mathcal{N}_\mu},
\end{equation}
where  $A_{\mathcal{N}_\mu} $ acts on ${\mathcal{N}_\mu} $.

\begin{theorem}
For a finite or compact Lie group  G,  let $\{g\rightarrow U^{(\mu)}(g)\}$  be the set of all inequivalent projective unitary irreps which are in the same factor system. Consider the matrix elements of all these unitary  matrices as a set of functions from G to $\mathbb{C}$ denoted by $\{U_{i,j}^{(\mu)}\}$. Then, they satisfy the following orthogonality relations
\begin{equation} \label{orthogonality}
\int_G dg\ U_{i,j}^{(\mu)}(g) {{\overline{U}^{(\nu)}_{k,l}}}(g)=\frac{\delta_{\mu,\nu}\delta_{i,k}\delta_{j,l}}{ d_\mu}
\end{equation}
where $dg$ is the unique Haar measure over the group, bar denotes the complex conjugate and $d_\mu$ is the dimension of irrep $\mu$. Furthermore, in the case of finite groups any function from G to $\mathbb{C}$ can be expanded as a linear combination of these functions. Also, in the case of compact Lie groups  any continuous function from G to $\mathbb{C}$ can be uniformly approximated as a linear combination of these matrix elements.
\end{theorem}
This expansion of functions in terms of the matrix elements of projective unitary irreps is called the \emph{generalized Fourier transform}. Note that for each cocycle of a group G there exists  a notion of generalized Fourier transform in which the functions over the group are expanded in terms of the matrix elements of the projective unitary irreps which all have that cocycle, and therefore are all in the same factor system.
As we have defined above, (non-projective) unitary representations are a specific case of projective unitary representations for which the cocycle is trivial.
So in particular, for any compact Lie group or finite group there is a unique generalized Fourier transform which corresponds to the (non-projective) unitary irreps of the group, i.e. the irreps for which the cocycle is trivial.

 In many cases the cocycle of a projective unitary representation can be \emph{lifted} in the sense that one can redefine the unitaries $\{U(g): g\in G\}$ by multiplying them by a phase such that the new unitaries form a (non-projective) unitary representation of the group and so the cocycle will be trivial. This is the case for all finite-dimensional representations of simply connected Lie groups such as SU(2), the group of unitaries acting on $\mathbb{C}^{2}$ with determinant one.\footnote{To see this, first note that by redefining the cocycle one can always choose the unitaries $\{U(g)\}$ to have determinant equal to one.  Then by looking at the determinant of both sides of Eq.~(\ref{Eq.cocycle}), one finds that for all $g_{1},g_{2}\in G$, it holds that $\omega^{d}(g_{1},g_{2})=1$ where $d$ is the dimension of the representation and so the values of $\omega(g_{1},g_{2})$ are discrete. Then using a simple continuity argument one can show that in the case of simply connected Lie  groups the cocycle $\omega(g_{1},g_{2})$ should be constant and equal to one and so the cocycle can be lifted.}. On the other hand, for Lie groups which are not simply connected, such as SO(3), the cocycle cannot always be lifted. This is the case for all irreps of SO(3) with  half-integer spin; they all have the same cocycle and this cocycle cannot be lifted. But, on the other hand,  for all irreps of SO(3) with  integer spin the cocycle is trivial and so they are all unitary irreps of SO(3).

This discussion implies that in the case of SO(3) there are two different notions of Fourier transform: One for the basis formed by the matrix elements of half-integer spin representations and the other for integer spin representations.

\section{Input-output Hilbert spaces} \label{embeding}

In general the input and output Hilbert space of a time evolution are not the same $(\mathcal{H}_\textnormal{in}\neq\mathcal{H}_\textnormal{out})$. This can happen especially in the case of open-system time evolutions.  However, we can always assume that the input and output spaces are two different sectors of  a larger  Hilbert space  ($\mathcal{H}_\textnormal{in}\oplus \mathcal{H}_\textnormal{out}$) and extend the time evolution to a  time evolution which acts on this larger Hilbert space. Therefore without loss of generality we can restrict our attention to the cases where the input and output Hilbert spaces are the same.

On the other hand, when the spaces are equipped with a representation of a symmetry group and the time evolution is covariant we may also care about the symmetries of time evolution of the extended system and therefore this process  of embedding spaces in a larger space is less trivial. Suppose there is a representation of group $G$ on the input and output Hilbert spaces given by $\{U_\textnormal{in}(g):g\in G\}$ and $\{U_\textnormal{out}(g):g\in G\}$. Suppose the time evolution is G-covariant, i.e., $\mathcal{E}\circ \mathcal{U}_\textnormal{in}(g)=\mathcal{U}_\textnormal{out}(g) \circ \mathcal{E}$ for all $g\in G$. In the following we will show that it is always possible to extend this time evolution to a time evolution on  $\mathcal{H}_\textnormal{in}\oplus \mathcal{H}_\textnormal{out}$ such that this extended time evolution respects the natural representation of $G$ on $\mathcal{H}_\textnormal{in}\oplus \mathcal{H}_\textnormal{out}$ given by $\{U_\textnormal{in}(g)\oplus U_\textnormal{out}(g):g\in G\}$. Therefore without loss of generality we can always restrict our attention to the G-covariant time evolutions whose input and output Hilbert spaces are the same. In particular when we ask whether there exists a G-covariant time evolution which maps $\rho$ to $\sigma$ we can always assume $\rho$ and $\sigma$ live in two sectors of the same Hilbert space.

\subsection{General G-covariant Channels}

 Suppose  $\mathcal{E}$ is a channel (completely-positive trace-preserving linear map) from $\mathcal{B}(\mathcal{H}_\textnormal{in})$ to $\mathcal{B}(\mathcal{H}_\textnormal{out})$ which is G-covariant, i.e.  for all $g\in G$ we have $U_\textnormal{out}(g)\mathcal{E}[\cdot]U^{\dag}_\textnormal{out}(g) =\mathcal{E}(U_\textnormal{in}(g)[\cdot]U^{\dag}_\textnormal{in}(g) )$. Then we can always extend this channel to $\tilde{\mathcal{E}}$, a G-covariant channel from  $\mathcal{B}(\mathcal{H}_\textnormal{in}\oplus \mathcal{H}_\textnormal{out})$ to itself, by defining
\begin{equation}
\tilde{\mathcal{E}}\equiv \mathcal{E}(\Pi_\textnormal{in}[\cdot] \Pi_\textnormal{in})+ \frac{   I_{\mathcal{H}_\textnormal{in}\oplus \mathcal{H}_\textnormal{out} }  }{d_\textnormal{in}+d_\textnormal{out} } \textnormal{tr}(\Pi_\textnormal{out}[\cdot] \Pi_\textnormal{out})
\end{equation}
where $ I_{\mathcal{H}_\textnormal{in}\oplus \mathcal{H}_\textnormal{out} } /({d_\textnormal{in}+d_\textnormal{out} })$ is the completely mixed state on $\mathcal{H}_\textnormal{in}\oplus \mathcal{H}_\textnormal{out}$. Clearly by this definition $\tilde{\mathcal{E}}$ is completely-positive and trace-preserving and so, a valid channel, and moreover it is G-covariant. Furthermore the restriction of $\tilde{\mathcal{E}}$ to $\mathcal{H}_\textnormal{in}$, i.e., $\tilde{\mathcal{E}}(\Pi_\textnormal{in}[ \cdot]\Pi_\textnormal{in} )$, is equal to $\mathcal{E}(\cdot)$.

On the other hand, if there is a G-covariant channel from $\mathcal{B}(\mathcal{H}_\textnormal{in}\oplus \mathcal{H}_\textnormal{out})$ to itself which maps all operators in  $\mathcal{B}(\mathcal{H}_\textnormal{in})$ to operators in $\mathcal{B}(\mathcal{H}_\textnormal{out})$ then clearly by restricting its input to $\mathcal{B}(\mathcal{H}_\textnormal{in})$  we get a valid G-covariant channel from  $\mathcal{B}(\mathcal{H}_\textnormal{in})$ to operators in $\mathcal{B}(\mathcal{H}_\textnormal{out})$.

Finally consider the situation where there is a G-covariant channel $\mathcal{E}$ from $\mathcal{B}(\mathcal{H})$ to itself which maps $\rho_{i}$ to $\sigma_{i}$ for a set of $i$'s. Assume the representation of the group $G$ on the Hilbert space is $\{U(g):g\in G\}$.  Define $\Pi_\textnormal{in}$ and $\Pi_\textnormal{out}$ to be respectively the span of the supports of all operators $\{U(g)\rho_{i} U^{\dag}(g)\}$ and $\{U(g)\sigma_{i} U^{\dag}(g)\}$. It is clear from this definition that both $\Pi_\textnormal{in}$ and $\Pi_\textnormal{out}$ commute with all $\{U(g):g\in G\}$. Therefore the subspaces associated to these  projectors,  $\mathcal{H}_\textnormal{in}$ and $\mathcal{H}_\textnormal{out}$, have a natural representation of the group $G$ given by $\{\Pi_\textnormal{in} U(g) \Pi_\textnormal{in} \}$ and $\{\Pi_\textnormal{out} U(g) \Pi_\textnormal{out}\} $. Now $\tilde{\mathcal{E}}\equiv \mathcal{E}(\Pi_\textnormal{in}[\cdot]\Pi_\textnormal{in})$ is a new G-covariant quantum channel which maps states from $\mathcal{B}(\mathcal{H}_\textnormal{in})$ to $\mathcal{B}(\mathcal{H}_\textnormal{out})$ and $\tilde{\mathcal{E}}(\rho_{i})=\sigma_{i}$.

\subsection{G-invariant unitaries and G-invariant isometries} \label{app-G-inv-isom}

 Basically we can repeat all of these observations to prove the equivalence of  a G-invariant unitary where the input and output spaces are the same and a G-invariant isometry where the input and output spaces are different.

For example if there exists a G-invariant unitary on $\mathcal{H}_\textnormal{in}\oplus \mathcal{H}_\textnormal{out}$ which unitarily maps the subspace $\mathcal{H}_\textnormal{in}$ to (a subspace of)   $\mathcal{H}_\textnormal{out}$ then clearly there exists a G-invariant isometry  $V$ from $H_\textnormal{in}$ to $\mathcal{H}_\textnormal{out}$ such that $\forall g\in G:\  VU_\textnormal{in}(g)=U_\textnormal{out}(g)V$ and $V^{\dag}V=I_\textnormal{in}$ where $I_\textnormal{in}$ is the identity on $\mathcal{H}_\textnormal{in}$.

The only property which is less trivial in the case of unitary-isometry equivalences is the following: Suppose $V$ is an isometry  from $H_\textnormal{in}$ to $H_\textnormal{out}$ which is G-invariant i.e.  $\forall g\in G: VU_\textnormal{in}(g)=U_\textnormal{out}(g)V$.   Then  there exits a unitary $V_{ext}$ on $\mathcal{H}_\textnormal{in}\oplus \mathcal{H}_\textnormal{out}$ such that  $\forall g:\in G: V_{ext} (U_\textnormal{in}(g)\oplus U_\textnormal{out}(g) )=(U_\textnormal{in}(g)\oplus U_\textnormal{out}(g) ) V_{ext}$ and moreover $V=\Pi_\textnormal{out}V_{ext}\Pi_\textnormal{in}$ where $\Pi_{in/out}$ is the projector to  $\mathcal{H}_{in/out}$.    This is shown by the following lemma
 \begin{lemma} \label{extension}
Suppose $W$  maps the subspace of the support of the projector $\Pi$ unitarily to another subspace such that   $\Pi W^{\dag}W\Pi=\Pi$ (in other words, $W\Pi$ is an isometry). Then if $\forall g\in G: [W\Pi,U(g)]=0$ there exits a unitary $W_{G-inv}$ such that $\forall g\in G: [W_{G-inv},U(g)]=0$ and $W_{G-inv}\Pi=W\Pi$.
\end{lemma}

\begin{proof}
$W\Pi$ commutes with all $U(g)$ and so does  $\Pi W^{\dag}$. Therefore $\Pi=\Pi W^{\dag}W\Pi$ also commutes with all $U(g)$.
Now we consider the decomposition of $U(g)$ to irreps,
\begin{equation}
U(g)=\bigoplus_{\mu} U_{\mu}(g) \otimes I_{\mathcal{N}_{\mu}}.
\end{equation}
Since  $\Pi$  commutes with all $\{U(g):g\in G\} $ it has a simple form in this basis:
\begin{equation}
\Pi=\bigoplus_\mu I_\mu\otimes {\Pi}^{(\mu)},
\end{equation}
where $\Pi^{2}=\Pi$ implies ${{\Pi}^{(\mu)}}^{2}={\Pi}^{(\mu)}$ and so all ${\Pi}^{(\mu)}$'s are projectors (Note that for some $\mu$, $\Pi_{\mu}$ might be zero.).  $W\Pi$ also commutes with all $\{U(g)\}$. Since $W\Pi=(W\Pi)\Pi$ we conclude that the decomposition of  $W{\Pi}$ should be in the following form
\begin{equation}
W\Pi=\bigoplus_\mu I_\mu\otimes (W^{(\mu)}{\Pi}^{(\mu)}).
\end{equation}
$\Pi W^{\dag}W\Pi=\Pi$ implies that ${\Pi}^{(\mu)}   {W^{(\mu)}}^\dag {W^{(\mu)}}{\Pi}^{(\mu)}={\Pi}^{(\mu)}$. Therefore ${W^{(\mu)}}{\Pi}^{(\mu)}$ acts unitarily on the subspace of the support of $\Pi^{(\mu)}$. Now we can always find a unitary $\tilde{W}^{(\mu)}$ on this subsystem such that $\tilde{W}^{(\mu)}{\Pi}^{(\mu)}=W^{(\mu)}{\Pi}^{(\mu)}$. Finally define the unitary $\tilde{W}$ as
\begin{equation}
{W}_{G-inv}=\bigoplus_\mu I_\mu\otimes \tilde{W}^{(\mu)}.
\end{equation}
Clearly it commutes with all $\{U(g)\}$ and $\tilde{W}\Pi=W\Pi$.
\end{proof}

\section{Characteristic functions and pairwise distinguishability} \label{app:distinguish}

In this section we discuss the interpretation of the amplitude of the characteristic function of $|\psi\rangle$  in terms of the pairwise distinguishability of states in the set $\{U(g)|\psi\rangle:g\in G\}$.

First, note that any measure of the distinguishability of a pair of pure states, $|\alpha_1\rangle$ and  $|\alpha_2\rangle$, depends only on the absolute value of their inner product, $|\langle\alpha_1|\alpha_2\rangle|$.  This is a consequence of the fact that for two pairs of states,  $\{ |\alpha_1\rangle\langle\alpha_{1}|,|\alpha_2\rangle\langle\alpha_{2}|\}$ and $\{ |\beta_1\rangle\langle\beta_{1}|,|\beta_2\rangle\langle\beta_{2}|\}$, the condition $|\langle\alpha_1|\alpha_2\rangle| = |\langle\beta_1|\beta_2\rangle|$ implies that it is possible, via a unitary dynamics, to reversibly interconvert between the two pairs, which in turn implies (on the grounds that no processing can increase the distinguishability of a pair of states) that they have the same distinguishability. Moreover using the same type of argument we can easily see that any measure of distinguishability should be monotonically nonincreasing in this overlap.  Therefore, for any pair of states $U(g_1)|\psi\rangle$ and $U(g_2)|\psi\rangle$, the distinguishability is specified by $|\langle \psi |U^\dag(g_1)U(g_2)|\psi\rangle|=|\chi_\psi(g_1^{-1}g_2)|$.

At first glance, therefore, one might think that the Gram matrix for any set of pure states merely encodes the distinguishability of every pair of these states, and therefore, that the characteristic function of a state merely encodes the pairwise distinguishability of every pair of elements in the group orbit of that state.
This is not the case however.  It is true that if two sets of states (in particular, two group orbits) are reversibly interconvertible (i.e., they have the same Gram matrix), then every pair from the first has the same distinguishability as the corresponding pair from the second. The opposite implication, however, fails.  In other words, the information content of the set  (in particular its entropy for different probability measures) is not specified by the pairwise distinguishabilities of its elements.

This phenomenon is highlighted by the results of Jozsa and Schlienz \cite{Jozsa-Schlienz}.  Also, a particularly nice example is provided by a result of Gisin and Popescu concerning the optimal state of two spin-half systems to use for sending a direction in space \cite{Gisin-Popescu}.
Define  $|{\uparrow}_{\hat{n}}\rangle$ and $|{\downarrow}_{\hat{n}}\rangle$ to be the eigenstates of spin along the $+\hat{n}$ direction, that is, $\hat{n}\cdot\vec{\sigma}|{\uparrow}_{\hat{n}}\rangle=|{\uparrow}_{\hat{n}}\rangle$ and $\hat{n}\cdot\vec{\sigma}|{\downarrow}_{\hat{n}}\rangle=-|{\downarrow}_{\hat{n}}\rangle$.
Then it is shown in \cite{Gisin-Popescu} that the state $\{|{\uparrow}_{\hat{n}}\rangle|{\downarrow}_{\hat{n}}\rangle\}$ is better than $\{|{\uparrow}_{\hat{n}}\rangle|{\uparrow}_{\hat{n}}\rangle\}$ for this task when the figure of merit is the fidelity of the estimated direction with the actual sent direction.  In other words, they showed that, with respect to this figure of merit, the  encoding $\{\Omega\rightarrow (U(\Omega)\otimes U(\Omega))|{\uparrow}_{\hat{z}}\rangle|{\downarrow}_{\hat{z}}\rangle, \Omega\in SO(3)\}$ provides more information  about $\Omega \hat{z}$ than the encoding  $\{\Omega\rightarrow(U(\Omega)\otimes U(\Omega))|{\uparrow}_{\hat{z}}\rangle|{\uparrow}_{\hat{z}}\rangle, \Omega\in SO(3)\}$.
On the other hand, one can easily check that the absolute values of the characteristic functions for the two states, which encode the pairwise distinguishability of elements of the orbits of the states, are exactly the same.
This follows from the fact that
\begin{align*}
|\chi_{\uparrow\downarrow}(\Omega)|&= \left|\ \langle{\uparrow}_{\hat{z}}|\langle{\downarrow}_{\hat{z}}|  \left[U(\Omega) \otimes U(\Omega)\right]|{\uparrow}_{\hat{z}}\rangle|{\downarrow}_{\hat{z}}\rangle\ \right|\\ &=|\langle{\uparrow}_{\hat{z}}| U(\Omega)|{\uparrow}_{\hat{z}}\rangle  |\times | \langle{\downarrow}_{\hat{z}}| U(\Omega)|{\downarrow}_{\hat{z}}\rangle|
\end{align*}
and
\begin{align*}
|\chi_{\uparrow\uparrow}(\Omega)|&= |\langle{\uparrow}_{\hat{z}}|\langle{\uparrow}_{\hat{z}}| \left[U(\Omega) \otimes U(\Omega)\right] |{\uparrow}_{\hat{z}}\rangle|{\uparrow}_{\hat{z}}\rangle|\\ &=|\langle{\uparrow}_{\hat{z}}| U(\Omega)|{\uparrow}_{\hat{z}}\rangle  |\times | \langle{\uparrow}_{\hat{z}}| U(\Omega)|{\uparrow}_{\hat{z}}\rangle|
\end{align*}
and the fact that for arbitrary rotation $\Omega$ we have $| \langle{\uparrow}_{\hat{z}}| U(\Omega)|{\uparrow}_{\hat{z}}\rangle|= | \langle{\downarrow}_{\hat{z}}| U(\Omega)|{\downarrow}_{\hat{z}}\rangle| $.

The insufficiency of the pairwise overlaps within a set of states for specifying the information contained in that  set implies that the relevant global properties of the set are encoded in the \emph{phases} of the components of the Gram matrix, or equivalently, for group orbits, in the phase of the characteristic function of the state generating the orbit.

One may think that the insufficiency of   pairwise  distinguishabilities for specifying the content of a set is a uniquely quantum phenomenon, but this is not the case. A simple example (attributed to Peter Shor in Ref.~\cite{Jozsa-Schlienz}) shows that the phenomenon can also arise with sets of classical probability distributions.  Consider a discrete sample space with four elements, and the following two sets of probability distributions: $\{ (1/2,1/2,0,0), (1/2,0,1/2,0), (0,1/2,1/2,0) \}$ and $\{ (1/2,1/2,0,0), (1/2,0,1/2,0), (0,1/2,0,1/2) \}$.
The three distributions in each case are illustrated by the ``sausages'' in Fig.~\ref{Fig:SausageDiagrams}.
It is clear that the pairwise overlaps are the same for the two sets but that they are not reversibly interconvertible.\footnote{It should be noted that the existence of this classical analogue demonstrates that the phenomenon in question can be added to the long list of those which are not obvious if one adopts the view that quantum states are states of reality, but are both intuitive and natural if one adopts the view that quantum states are states of incomplete knowledge \cite{Spe04}.}

\begin{figure}[h!]
   \includegraphics[width=8cm]{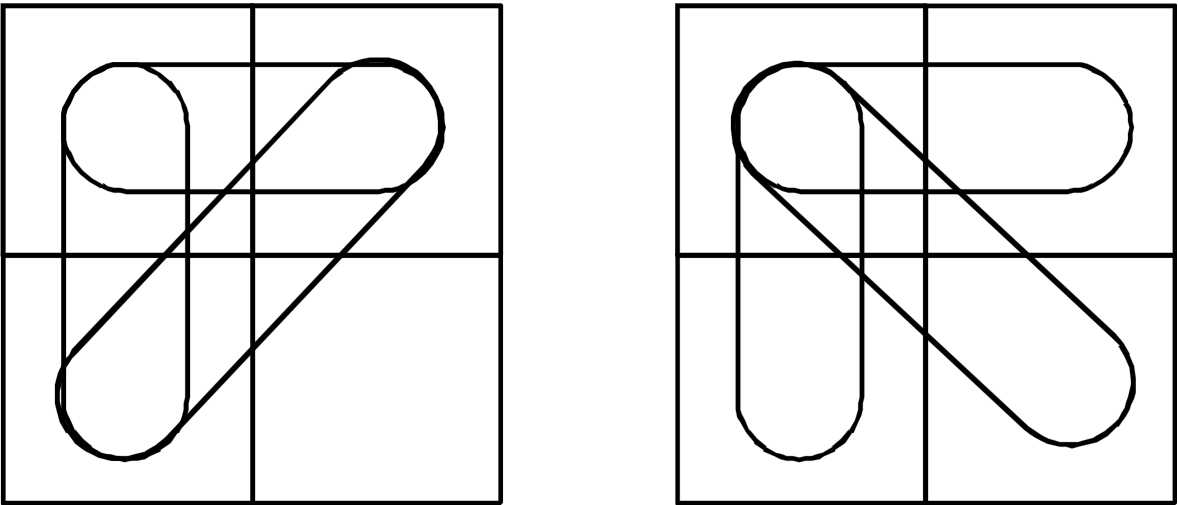}
    \caption{\label{Fig:SausageDiagrams}
    Example of two ensembles of classical probability distributions that have different information content, but for which the pairwise distinguishability are the same.
    }
\end{figure}



\section{Comparison of classical and quantum characteristic functions } \label{app:charfuncs}

The characteristic function of a quantum state can be understood as a generalization of the characteristic function of a probability distribution. In fact, this generalization was the first motivation for introducing the notion of a characteristic function for a quantum state by Gu \cite{Gu}. We first review some properties of classical characteristic functions and then we talk about their analogues in the case of quantum states and non-Abelian groups. We also review positive definiteness as the main criterion for a complex function over the group to be the characteristic function of a valid quantum state. Almost all the materials of this appendix are borrowed from  \cite{Gu,Ent-charac1, Ent-charac2}.

\subsection{Review of classical characteristic functions}

 For a real random variable $x$ with the distribution function $F(x)$  the characteristic function is defined as the expectation value of the random variable $e^{{i tx}}$ i.e.
\begin{equation}
f_{x}(t)=\int dF(x)   e^{itx}
\end{equation}
The distribution function is uniquely determined by its characteristic function. Moreover if the probability density exists then it will be equal to the inverse Fourier transform of the characteristic function. One particularly useful property of the characteristic function is the multiplicative property according to which the characteristic function of the sum of two independent random variables is equal to the product of their characteristic functions.
\begin{equation}
f_{x+y}(t)=f_{x}(t)f_{y}(t)
\end{equation}
There exists a remarkably simple proof of the central limit theorem using this multiplicative property of characteristic functions.

The derivative of characteristic functions at the origin determines the moments of the random variable.
\begin{equation}
\langle x^n \rangle =  i^{-n} \frac{d^{n}}{dt^{n}} f_{x}(t)\mid_{t=0}
\end{equation}

Sometimes it is more favourable to use \emph{cumulants} of  the random variable instead where  the $n$-th order {cumulant} is defined as the $n$-th order derivative of the logarithm of the characteristic function at the point $0$, multiplied by $i^{-n}$.
\begin{equation}
\kappa^{(n)}\equiv i^{-n} \frac{d^{n}}{dt^{n}} \log(f_{x}(t))\mid_{t=0}
\end{equation}
The first and second cumulants are the mean and the variance of the random variable.  By this definition, it turns out that  cumulants of a sum of independent random variables is equal to the sum of the cumulants of the individual terms for all orders of cumulants.

The set of all classical characteristic functions is determined by Bochner's theorem, according to which a complex function $f(t)$ is the characteristic function of a random variable if and only if (i) $f(0) = 1$,  (ii) $f(t)$ is continuous at the origin, and(iii) it is \emph{positive definite}. Recall  that a function $f(t)$ is positive definite if for any integer $n$ and for any string of real  numbers $t_{1},...,t_{n}$ the matrix $a_{i,j}\equiv f(t_{i}-t_{j})$ is a positive definite matrix.  Positive definiteness of a function guarantees that the inverse Fourier transform of  this function is positive for all values of the random variable, which is clearly a necessary condition for a function to be a probability density.

For more discussion about the properties of characteristic functions of probability distributions, see e.g. \cite{Gnedenko}.

\subsection{Quantum characteristic functions} \label{app:qcharfuncs}

As the characteristic function of a probability distribution determines all of its statistical properties, the characteristic function of a quantum state over the group $G$ uniquely specifies all the statistical properties of observables in the algebra of observables which generates the projective unitary representation of $G$. For example, suppose $L$ is the representation of a generator of the Lie  group $G$, then we have
\begin{equation}
\textnormal{tr}(\rho L^{k})= i^{-k}\frac{\partial^{k}}{\partial \theta^{k}} \chi_{\rho}(e^{i\theta L}) \mid_{{{\theta=0}}}
\end{equation}
In particular the first derivative $(k=1)$ determines the expectation value of the generator. This is just property \ref{derivatives} of characteristic functions from Section \ref{properties}.

  Similarly we can define \emph{cumulants} of  the observable $L$, where  the $n$-th order {cumulant} is defined as the $n$-th order derivative of the logarithm of the characteristic function at the identity element multiplied by $i^{-n}$.
\begin{equation}
\kappa_{L}^{(n)}\equiv i^{-k}\frac{\partial^{k}}{\partial \theta^{k}} \log[\chi_{\rho}(e^{i\theta L})] \mid_{{{\theta=0}}}
\end{equation}
The first and second cumulants are the mean and the variance of the observable.  By this definition, it turns out that the cumulants of the tensor product of two states  is equal to the sum of the cumulants of the individual states for all orders of cumulants.
\color{black}

In the rest of this appendix, we are interested to find the generalization of Bochner's theorem, i.e., the set of necessary and sufficient conditions for $\phi(g)$,  a complex function over group, to be the characteristic function of some quantum state.
We see that such a generalization can be found via both the non-commutative Fourier transform and the Gelfand-Naimark-Segal (GNS) construction theorem.  As in the rest of the paper, we focus on the finite groups and compact Lie groups.

As the first necessary condition  we note that   $\textnormal{tr}(\rho)=1$  implies that  $\chi(e)=1$ (where $e$ is the identity of the group). We call the functions which satisfy this condition \emph{normalized} functions. In the case of compact Lie groups, $\phi(g)$ should also be a continuous function. We also need a condition on $\phi(g)$ equivalent to the positivity of density operators. As we just saw in the case of probability distributions the condition of positivity of probabilities is equivalent to the positive definiteness of the characteristic function of the probability distribution. Similarly it turns out that the relevant condition on  $\phi(g)$ to be the characteristic function of a positive operator is the natural generalization of positive definiteness for the functions defined on the group:
\begin{definition}
A complex function $\phi(g)$ on a group $G$ is positive definite if for all choices $m\in \mathbb{N}$, $g_{1},...g_{m}\in G$ and $\alpha_{1},...\alpha_{m}\in \mathbb{C}$
\begin{equation}
\sum_{{i,j=1}}^{m} \bar{\alpha_{i}}\alpha_{j} \phi(g^{-1}_{i}g_{j}) \ge 0
\end{equation}
\end{definition}
For the case of compact Lie groups where the function should also be continuous we can express the condition as
\begin{definition}
A  continuous function $\phi(g)$ on a group G with the Haar measure $dg$ is called positive definite if it satisfies
\begin{equation} \label{Def-Pos}
\int\int dgdh\ \bar{f}(g)\phi(g^{-1}h)f(h)\geq\ 0
\end{equation}
for any $f\in L^{1}(G)$.
\end{definition}
Now using the Fourier transform, one can easily prove a theorem similar to Bochner's theorem \cite{Gu,Ent-charac1}:
\begin{theorem}
A complex function $\phi(g)$ on the finite or compact Lie  group $G$ is the characteristic function of a quantum state in a finite-dimensional Hilbert space iff $\phi(e)=1$,  $\phi(g)$ is positive definite and continuous (in the case of Lie groups).
\end{theorem}
\color{black}
\begin{proof}
We present the proof assuming that the group $G$ is a compact Lie group (The same argument works for a finite group by replacing integrals with summation.).  We use the inverse Fourier transform. Suppose $B^{(\mu)}\equiv d_{\mu}\int dg U^{(\mu)}(g^{-1})\phi(g)$. Then the set of operators $\{B^{(\mu)}\}$ is the reduction onto irreps of a valid quantum state  iff (1) $\sum_{\mu} \textnormal{tr}(B^{(\mu)})=1$ and (2) all operators $\{B^{(\mu)}\}$ are positive definite. The first condition expresses the fact that the trace of the state is 1 and is guaranteed by $\phi(e)=1$.  On the other hand,  $B^{(\mu)}$ is positive iff $\textnormal{tr}(FF^\dag B^{(\mu)})\geq 0$ for all operators $F$ acting on $\mathcal{M}_\mu$ (the subsystem on which $U^{\mu}$ acts irreducibly). Note that $\textnormal{tr}(FF^\dag B^{(\mu)})$ is equal to the Fourier transform of  the operator $FF^\dag B^{(\mu)}$ at point $e$. So using the convolution property of characteristic functions, Eq.(\ref{convolution}), we get
\begin{equation}
\textnormal{tr}(FF^\dag B^{(\mu)})=d^2_\mu \int \int dh_1 dh_2 f(h_1) \overline{f(h_2)} \phi(h^{-1}_1h_2).
\end{equation}
So if $\phi(g)$ is positive definite and therefore satisfies Eq.~(\ref{Def-Pos}) then all $B^{(\mu)}$'s are positive. We can prove the other direction of the theorem similarly.
\end{proof}

Therefore the set of normalized  positive definite functions (also continuous in the case of Lie groups) are exactly the set of characteristic functions of states.

We can also get this result using a more fundamental theorem in the representation theory of $C^*$ algebras, namely, the GNS construction. A specific form of this theorem states
\begin{theorem}[GNS construction]
With every (continuous) positive definite function $\phi(g)$ we can associate a Hilbert space $\mathcal{H}$, a unitary representation $\{U(g): g\in G\}$ of $G$ in $\mathcal{H}$ and a vector $\psi$, cyclic for $\{U(g): g\in G\}$, such that
\begin{equation}
\phi(g)=\left\langle \psi\right\vert U(g)\left\vert \psi\right\rangle .
\end{equation}
Moreover the representation $\{U(g)\}$ is unique up to a unitary equivalence.
\end{theorem}

Note that a vector $|\xi\rangle$ is cyclic for the representation $\{U(g): g\in G\}$ on the space $\mathcal{H}$  if the span of  vectors $\{U(g)|\xi\rangle:\ g\in G\}$ is a dense subset of the space $\mathcal{H}$.

Therefore the GNS construction theorem guarantees that for any given (continuous) normalized positive definite function there exists a corresponding pure cyclic state with that characteristic function.
Note that for any arbitrary mixed or pure state there exists a pure state which is cyclic (for the representation on its Hilbert space) with exactly the same characteristic function. So the set of all (continuous) normalized, positive definite functions is exactly the same as  the set of all characteristic functions of states.


\section{More on the approximate notion of unitary G-equivalence} \label{app:proofofapproxGequivalence}

In this section,  we prove  theorem~\ref{approximate_theorem} and   present some other versions of this result.

Using the standard bounds between fidelity and trace distance of two operators  \cite{Nie00}, we can express this result in terms of the trace distance between the reductions.   As it may be useful in future applications, we present this reformulation of the condition as a corollary of theorem~\ref{approximate_theorem}.
\begin{corollary} \label{cor:closenessofreductions}
Suppose  $\{F_1^{(\mu)}\}$ and $\{F_2^{(\mu)}\}$ are respectively the reductions onto irreps of states $\psi_1, \psi_2\in\mathcal{H}$. Then there exists a  G-invariant unitary $V$  acting on $\mathcal{H}$  such that
\begin{equation}
|\langle\psi_2|V|\psi_1\rangle| \ge 1-\frac{1}{2}\sum_\mu \|F_{1}^{(\mu)}-F_{2}^{(\mu)}\|.
\end{equation}
\end{corollary}

In the following we present a similar bound in terms of the distance between characteristic functions of states  $\chi_{\psi_{1,2}}(g)$ and another bound in terms of the distance between the components of characteristic functions $\{\chi^{(\mu)}_{\psi_{1,2}}(g)\}$ where the $\mu$ component of  $\chi_{\psi_{1,2}}(g)$ is defined as
\color{black}
\begin{align*}
\chi^{(\mu)}_{\psi_{1,2}}(g)&\equiv \textnormal{tr}(U^{(\mu)}(g)F_{1,2}^{(\mu)})\\
&= d_{\mu}\ \textnormal{tr}(U^{(\mu)}(g) \int dh U^{(\mu)}(h^{-1}) \chi_{\psi_{1,2}}(h)  )\\
 &=d_{\mu}\  \left( \varphi_{\mu}\ast\chi_{\psi_{1,2}}\right) (g),
\end{align*}
where $\varphi_{\mu}(g)=\textnormal{tr}(U^{(\mu)}(g))$ is the character of irrep $\mu$ and $*$ is the convolution operation defined in Eq.~(\ref{def-conv}).

\begin{corollary} \label{cor:closenessintermsofcharfuncs}
Suppose  $\chi_{\psi_1}$ and $\chi_{\psi_2}$ are respectively the characteristic functions of states $\psi_1$ and $\psi_2$. Then there exists a  G-invariant unitary $V$  such that
\begin{equation}
|\langle\psi_2|V|\psi_1\rangle|\geq 1-\frac{1}{2} (\sum_\mu d^2_\mu) \int dg |\chi_{\psi_{1}}(g)-\chi_{\psi_{2}}(g)|,
\end{equation}
and
\begin{equation}
|\langle\psi_2|V|\psi_1\rangle|\geq 1-\frac{1}{2} \sum_\mu d^2_\mu \left( \int dg\  |\chi^{(\mu)}_{\psi_{1}}(g)-\chi^{(\mu)}_{\psi_{2}}(g)|\right),
\end{equation}
where the summation is over all irreps in which  $\psi_1$ and $\psi_2$ have nonzero components.
\end{corollary}

To prove theorem \ref{approximate_theorem}, we first recall a well-known theorem by Uhlmann (see e.g. \cite{Watrous}).
\begin{theorem}
\textbf{(Uhlmann)} Suppose $A_1$ and $A_2$ are two positive operators on $\mathcal{H}$. Also suppose $\mathcal{H'}$ is a space large enough such that $\mathcal{H}\otimes\mathcal{H}'$ admits a purification of both $A_1$ and $A_2$. Suppose for $k\in\{1,2\}$ that $|\alpha_k\rangle$ is a purification of $A_k$ on  $\mathcal{H}\otimes\mathcal{H'}$, i.e. $tr_{\mathcal{H'}}(|\alpha_k\rangle\langle\alpha_k|)=A_{k}$.  In this case,
\begin{align}
\textnormal{Fid}(A_1,A_2) &\equiv \|\sqrt{A_1}\sqrt{A_2}\| \\ &=\max\{|\langle\alpha_1|\alpha_2\rangle| :  tr_{\mathcal{H'}}(|\alpha_2\rangle\langle\alpha_2|)=A_2 \}.
\end{align}
\end{theorem}


\begin{proof} (theorem \ref{approximate_theorem} and remark \ref{remark-unit})

Suppose $\mathcal{M}_\mu\otimes \mathcal{N}_\mu$ is the  subspace associated to irrep $\mu$ in $\mathcal{H}$ and $\Pi_\mu$ is the projective operator to this subspace. Define
\begin{equation}
|\psi^{(\mu)}_{1,2}\rangle\equiv \Pi_\mu |\psi_{1,2}\rangle.
\end{equation}
Suppose $V$ is an arbitrary G-invariant unitary. Define  $|\widetilde{\psi}\rangle\equiv V|\psi_1\rangle$ and $ |\widetilde{\psi}^{(\mu)}\rangle\equiv  \Pi_\mu  V|\psi_1\rangle$.  Then
\begin{equation} \label{bound1}
|\langle\psi_2|V|\psi_1\rangle|=|\sum_\mu \langle\psi^{(\mu)}_2|\widetilde{\psi}^{(\mu)}\rangle| \le \sum_\mu |\langle\psi^{(\mu)}_2|\widetilde{\psi}^{(\mu)}\rangle|.
\end{equation}

Also define
\begin{equation}
F_{1,2}^{(\mu)}\equiv \textnormal{tr}_{\mathcal{N_\mu}}(|\psi^{(\mu)}_{1,2}\rangle\langle\psi^{(\mu)}_{1,2}|),
\end{equation}
where $F_{1}^{(\mu)}$ and $F_{2}^{(\mu)}$ are both operators acting on $\mathcal{M}_\mu$.

The fact that $V$ is G-invariant implies that $|\widetilde{\psi}\rangle$ and   $|{\psi_{1}}\rangle$ have the same reductions onto irreps, i.e., for all $\mu$
\begin{equation}
\textnormal{tr}_{\mathcal{N_\mu}}(|\widetilde{\psi}^{(\mu)}\rangle\langle\widetilde{\psi}^{(\mu)}|)=\textnormal{tr}_{\mathcal{N_\mu}}(|\psi^{(\mu)}_{1}\rangle\langle\psi^{(\mu)}_{1}|)=F_{1}^{(\mu)}.
\end{equation}
Since $|\widetilde{\psi}^{(\mu)}\rangle$ and $|\psi^{(\mu)}_2\rangle$ are purifications of $F_{1}^{(\mu)}$ and $F_{2}^{(\mu)}$, then according to Uhlmann's theorem,
\begin{equation}
|\langle\psi^{(\mu)}_2|\widetilde{\psi}^{(\mu)}\rangle| \le \textnormal{Fid}(F_{1}^{(\mu)},F_{2}^{(\mu)}).
\end{equation}
This inequality together with the inequality (\ref{bound1}) implies the bound (\ref{approximate_trans}).

Now we prove that this bound is achievable. According to Uhlmann's theorem there exists a purification of $F_{1}^{(\mu)}$, denoted by $|{\phi}^{(\mu)}\rangle$, such that
\begin{equation}\label{eq-proof-ineq-fid}
\textnormal{Fid}(F_{1}^{(\mu)},F_{2}^{(\mu)})= |\langle \psi_2^{(\mu)}|{\phi}^{(\mu)}\rangle|.
\end{equation}
But all purifications of $F_{1}^{{(\mu)}}$ can be transformed to each other by unitaries acting on  $\mathcal{N}_\mu$ (and acting trivially on $\mathcal{M}_\mu$). So there exists a unitary $V^{(\mu)}$ acting on $\mathcal{N_\mu}$ such that  $I\otimes V^{(\mu)} |\psi_1^{(\mu)}\rangle=|{\phi}^{(\mu)}\rangle$.
Now define
\begin{equation}
V\equiv \bigoplus_\mu e^{i\theta_\mu} I\otimes V^{(\mu)}
\end{equation}
where $\{e^{i\theta_\mu}\}$ are chosen such that all the numbers $\{e^{i\theta_\mu}\langle \psi_2^{(\mu)}|{\phi}^{(\mu)}\rangle\}$ have the same phase. Note that with this definition $V$ is a G-invariant unitary. Then we get
\begin{equation}
|\langle\psi_2|V|\psi_1\rangle|=|\sum_\mu e^{i\theta_\mu}  \langle \psi_2^{(\mu)}|{\phi}^{(\mu)}\rangle|= \sum_\mu |\langle \psi_2^{(\mu)}|{\phi}^{(\mu)}\rangle|
\end{equation} 
where the second equality holds because we have chosen $\{e^{i\theta_\mu}\}$ such that all the terms in the summand  have the same phase. Therefore for this G-invariant unitary we have
\begin{equation} \label{bound}
|\langle\psi_2|V|\psi_1\rangle|=\sum_\mu \textnormal{Fid}(F_{1}^{(\mu)},F_{2}^{(\mu)}).
\end{equation}
This completes the proof of theorem \ref{approximate_theorem}. To prove remark \ref{remark-unit}, we infer from Eq.~(\ref{eq-proof-ineq-fid}) that
\begin{align*}
\sum_{\mu}\textnormal{Fid}(F_{1}^{(\mu)},F_{2}^{(\mu)})&= \sum_{\mu}|\langle \psi_2^{(\mu)}|{\phi}^{(\mu)}\rangle|\\ &\le \sum_{\mu}\sqrt{\langle \psi_2^{(\mu)}|\psi_2^{(\mu)}\rangle} \sqrt{\langle {\phi}^{(\mu)}|{\phi}^{(\mu)}\rangle}\\ &\le \sqrt{\sum_{\mu}{\langle \psi_2^{(\mu)}|\psi_2^{(\mu)}\rangle}} \sqrt{\sum_{\mu} {\langle {\phi}^{(\mu)}|{\phi}^{(\mu)}\rangle}}=1,
\end{align*}
where both of the inequalities are implied by the Cauchy-Schwarz inequality and the last equality is implied by the normalization of states. Now we note that the last inequality holds as an equality iff $\forall\mu: \ \langle \psi_2^{(\mu)}|\psi_2^{(\mu)}\rangle=k \langle \phi^{(\mu)}|\phi^{(\mu)}\rangle$ for some constant $k$. But the normalization of states implies that $\forall \mu: \langle \psi_2^{(\mu)}|\psi_2^{(\mu)}\rangle=\langle \phi^{(\mu)}|\phi^{(\mu)}\rangle=1$.  Furthermore, the first inequality holds as an equality  if and only if for each $\mu$ there is a constant $c_{\mu}$ such that $|\psi_2^{(\mu)}\rangle=c_{\mu}|\phi^{(\mu)}\rangle$. These two observations together imply that $\sum_{\mu}\textnormal{Fid}(F_{1}^{(\mu)},F_{2}^{(\mu)})\le 1$ and the equality holds only if
\begin{equation}
\forall\mu:\ |\psi^{(\mu)}_{2}\rangle\langle\psi^{(\mu)}_{2}|=|\phi^{(\mu)}\rangle\langle\phi^{(\mu)}|.
\end{equation}
But $|\psi^{(\mu)}_{2}\rangle$ is a purification of $F_{2}^{(\mu)}$ and $|\phi^{(\mu)}\rangle$ is a purification of $F_{1}^{(\mu)}$. So the above equality implies that
\begin{equation}
\forall\mu:\ F_{1}^{(\mu)}=F_{2}^{(\mu)}.
\end{equation}
This completes the proof of remark~\ref{remark-unit}.
\end{proof}

To prove corollary \ref{cor:closenessofreductions}, we begin by recalling some facts about the trace distance.
For density operators $\rho_{1}$ and $\rho_{2}$ it is well known  that  $\|\rho_{1}-\rho_{2}\|\ge 2(1- \textnormal{Fid}(\rho_1,\rho_2))$ \cite{Nie00,Watrous}. Using the same argument it can be easily seen that for general positive operators $A_{1}$ and $A_{2}$,  we have the following lemma
\begin{lemma} \label{lemma:tracedistancefidelity}
Suppose  $A_1$ and $A_2$ are two positive operators. Then
\begin{equation}
\|A_1-A_2\| \ge \textnormal{tr}(A_1)+\textnormal{tr}(A_2)-2 \textnormal{Fid}(A_1,A_2).
\end{equation}
\end{lemma}
We now provide the proof.

\begin{proof}(corollary \ref{cor:closenessofreductions})

According to lemma \ref{lemma:tracedistancefidelity},
\begin{equation}
\textnormal{Fid}(F_{1}^{(\mu)},F_{2}^{(\mu)}) \ge \frac{1}{2}\left(\textnormal{tr}(F_{1}^{(\mu)})+\textnormal{tr}(F_{2}^{(\mu)})-\|F_{1}^{(\mu)}-F_{2}^{(\mu)}\|\right),
\end{equation}
which implies
\begin{align*}
\sum_\mu \textnormal{Fid}(F_{1}^{(\mu)},F_{2}^{(\mu)})  \ge \frac{1}{2}&(\sum_\mu \textnormal{tr}(F_{1}^{(\mu)})
+\\ &\sum_\mu \textnormal{tr}(F_{2}^{(\mu)})- \sum_\mu\|F_{1}^{(\mu)}-F_{2}^{(\mu)}\|)\\
&=1-\frac{1}{2}\sum_\mu\|F_{1}^{(\mu)}-F_{2}^{(\mu)}\|,
\end{align*}
where we have used the fact that the sum of the traces of the elements of the reduction onto irreps is 1.  Combining this bound with theorem \ref{approximate_theorem}, we obtain the desired result.
\end{proof}

\begin{proof} (corollary \ref{cor:closenessintermsofcharfuncs})
According to the Fourier transform, Eq.~(\ref{Fourier}),
\begin{equation}
F_{1,2}^{(\mu)}=d_\mu\int dg\  U^{(\mu)}(g^{-1}) \chi_{\psi_{1,2}}(g).
\end{equation}
Therefore
\begin{align*}
\|F_{1}^{(\mu)}-F_{2}^{(\mu)}\|&=d_\mu\  \|\left|\int dg\ U^{(\mu)}(g^{-1}) \left[\chi_{\psi_{1}}(g)-\chi_{\psi_{2}}(g)\right]\right|\|\\ &\le  d_\mu \int dg\   \|U^{(\mu)}(g^{-1})\| \left| \chi_{\psi_{1}}(g)-\chi_{\psi_{2}}(g) \right|.
\end{align*}
Since $U^{(\mu)}(g^{-1})$ is a unitary acting on a $d_\mu$-dimensional space, $\|U^{(\mu)}(g^{-1})\|=d_\mu$. So we have
\begin{equation}
\|F_{1}^{(\mu)}-F_{2}^{(\mu)}\| \le  d^2_\mu \int dg \left|\chi_{\psi_{1}}(g)-\chi_{\psi_{2}}(g)\right|.
\end{equation}
Therefore we have
\begin{equation}
\sum_\mu \|F_{1}^{(\mu)}-F_{2}^{(\mu)}\| \le  \left(\sum_\mu d^2_\mu\right) \int dg |\chi_{\psi_{1}}(g)-\chi_{\psi_{2}}(g)|,
\end{equation}
where the summation is over all irreps in which  $\psi_1$ and $\psi_2$ have nonzero components.


The second bound on $\sum_\mu \|F_{1}^{(\mu)}-F_{2}^{(\mu)}\|$ is obtained as follows.

Recalling the definition of the $\mu$ component of  $\chi_{\psi_{1,2}}(g)$,
the orthonormality of matrix elements of different irreps implies
\begin{equation}
F_{1,2}^{(\mu)}=d_\mu\int dg\  U^{(\mu)}(g^{-1}) \chi^{(\mu)}_{\psi_{1,2}}(g).
\end{equation}
Therefore
\begin{align*}
\|F_{1}^{(\mu)}-F_{2}^{(\mu)}\|&= d_\mu\|\int dg\ U^{(\mu)}(g^{-1}) \left[\chi^{(\mu)}_{\psi_{1}}(g)-\chi^{(\mu)}_{\psi_{2}}(g)\right]\| \\ &\le  d_\mu \int dg\   \|U^{(\mu)}(g^{-1})\| \left| \chi^{(\mu)}_{\psi_{1}}(g)-\chi^{(\mu)}_{\psi_{2}}(g) \right|.
\end{align*}
Using the fact that $\|U^{(\mu)}(g^{-1})\|=d_\mu$ again, we have
\begin{equation}
\|F_{1}^{(\mu)}-F_{2}^{(\mu)}\| \le  d^2_\mu \int dg \left|\chi^{(\mu)}_{\psi_{1}}(g)-\chi^{(\mu)}_{\psi_{2}}(g)\right|.
\end{equation}
Therefore we have
\begin{equation}
\sum_\mu \|F_{1}^{(\mu)}-F_{2}^{(\mu)}\| \le  \sum_\mu d^2_\mu \int dg\  |\chi^{(\mu)}_{\psi_{1}}(g)-\chi^{(\mu)}_{\psi_{2}}(g)|,
\end{equation}
where the summation is over all irreps $\mu$ in which  $F_{1}^{(\mu)}$ or $F_{2}^{(\mu)}$ are nonzero.
\end{proof}


\begin{thebibliography}{000}

\bibitem{Goldstein} H. Goldstein, \textit{Classical Mechanics}, 2nd ed. (Reading, MA: Addison-Wesley Publishing,  1980).

\bibitem{BRS07} S. D. Bartlett, T. Rudolph, and R. W. Spekkens, Rev. Mod.
Phys. {79}, 555 (2007).

\bibitem{GS07}
G.~Gour, and R.~W.~Spekkens, New J. Phys. {10}, 033023 (2008), quant-ph/0711.0043v2.


\bibitem{GMS09}  G. Gour, I. Marvian, and R. W. Spekkens, Phys. Rev. A 80, 012307 (2009),  quant-ph/0901.0943v2.

\bibitem{Nie00}
M. A. Nielsen, and I. L. Chuang, \textit{Quantum Computation and Quantum Information}, (Cambridge University Press, Cambridge, 2000).

\bibitem{Nielsen-Ent} M. A. Nielsen, Phys. Rev. Lett. 83, 436 (1999), quant-ph/9811053.


\bibitem{Wigner}
E.~P.~Wigner, \textit{Group Theory} (Academic Press Inc., New York, 1959), pp. 233-236.



\bibitem{Chiribella} G. Chiribella, \textit{Optimal estimation of quantum signals in the presence of symmetry}, (PhD thesis, University of Pavia, Pavia, Italy, 2006), \color{green}\textnormal{{\url{http://www.qubit.it/educational/thesis/ThesisRevised.pdf}}}.\color{black}


\bibitem{preparation} Chiribella, Marvian, Spekkens, in preparation.

\bibitem{thesis:Marvian} I. Marvian, \textit{Symmetry, asymmetry and quantum information}, (PhD thesis, University of Waterloo, Waterloo, Canada, 2012).

\bibitem{SVC04_PRA} N.~Schuch, F.~Verstraete, and J.~I. Cirac, Phys. Rev. A 70, 042310 (2004),  quant-ph/0404079.


\bibitem{SVC04_PRL} N.~Schuch, F.~Verstraete, and J.~I. Cirac, Phys. Rev. Lett. 92, 087904 (2004), quant-ph/0310124.


\bibitem{Vac-Wise-Jac} J. A. Vaccaro, F. Anselmi, H. M. Wiseman, and K. Jacobs, Phys. Rev. A 77, 032114 (2008).

\bibitem{Skot-Gour} M. Skotiniotis, and G. Gour,  New J. Phys. 14, 073022, quant/ph-1202.3163.


\bibitem{Tol-Gour-Sand} B. Toloui, G. Gour, and B. C. Sanders, Phys. Rev. A 84, 022322 (2011),  quant-ph/1104.1144.

\bibitem{Vac2012}  J. A. Vaccaro, Proc. R. Soc. A 468, 1065 (2012). 


\bibitem{Werner}
M. Keyl, and R. F. Werner, J. Math. Phys. \textbf{40}, 3283 (1999), quant-ph/9807010v1.

\bibitem{Curie} P. Curie, Journal de Physique \textbf{3}, 401 (1894).

\bibitem{Jozsa-Schlienz} R. Jozsa, and J. Schlienz, Phys. Rev. A {62}, 012301-1 (1999), quant-ph/9911009v1.

\bibitem{Davidson} K. Davidson, \textit{C*-algebras by example,  Fields Institute Monographs}, (Amer. Math. Soc., Providence, 1996).

\bibitem{Barut-1986} A. O. Barut and R. Raczka, \textit{Theory of Group Representations and Applications}, (World Scientific, 1986).


\bibitem{Spe04} R.~W.~Spekkens, Phys. Rev. A \textbf{75},
    032110 (2007); arXiv:quant-ph/0401052.




\bibitem{Jon-Pel} D. Jonathan, M. Plenio, Phys. Rev. Lett. 83, 3566(1999).


\bibitem{Gnedenko}  B.V.  Gnedenko, \textit{The theory of probability}, (Chelsa, New York, 1962).

\bibitem{Gisin-Popescu} N.~Gisin and S.~Popescu, Phys. Rev. Lett. {83}, 432 (1999), quant-ph/9901072v1.


\bibitem{Gu} Y. Gu,  Phys. Rev. A {32}, 1310 (1985).

\bibitem{Ent-charac1} J. K. Korbicz, M. Lewenstein, Phys. Rev. A {74}, 022318 (2006), quant-ph/0601189.
\bibitem{Ent-charac2} J. K. Korbicz, J. Wehr, M. Lewenstein, Com. Math. Phys. {281}, 753 (2008), quant-ph/0705.2965.

\bibitem{Watrous} John Watrous's lecture notes,
\color{green}\textnormal{{\url{http://www.cs.uwaterloo.ca/~watrous/lecture-notes.html}}}.\color{black}















\end{thebibliography}
\end{document}